%% file: main_inverseSFs.tex
\documentclass[a4paper,USenglish,cleveref, autoref, thm-restate,anonymous,envcountsame]{llncs}
\input{config}

\newif\ifpaper
\papertrue 
\paperfalse 

\usepackage{currfile}
\usepackage{lineno}
\usepackage{filehook}
\usepackage{xstring}
\usetikzlibrary{positioning, calc}

\usepackage{graphicx}

\title{(Sets of) Complement Scattered Factors} 

\titlerunning{Complement Scattered Factors} 

 \author{Duncan Adamson\inst{1} \and 
 Pamela Fleischmann\inst{2} \and
 Annika Huch\inst{2}}

 \institute{University of St Andrews, UK\\
 	\email{duncan.adamson@st-andrews.ac.uk}\and Kiel University, Germany, \email{$\{$fpa,ahu$\}$@informatik.uni-kiel.de}}

\authorrunning{D. Adamson \and P. Fleischmann, \and A. Huch}

\newtheorem{observation}[theorem]{Observation}

\begin{document}
%
	%
	
\maketitle

\input{abstract}


\section{Introduction}
\input{intro}

\section{Preliminaries}\label{sec:prelims}
\input{prelims}

\section{Computing $C(w,u)$}\label{algo}
\input{computingC}

\section{Complement Scattered Factors}\label{sec:csf}
\input{results}
\input{singletonEandCsets}
\input{selfshuffle}

\section{Conclusion}
\input{conclusion}




\bibliographystyle{plainurl}
\bibliography{refs}

\newpage

\appendix



\end{document}

%% file: config.tex
\usepackage[T1]{fontenc}
\usepackage{graphicx}
\usepackage{amssymb}
\usepackage{amsmath}
\usepackage{mathtools}
\usepackage{todonotes}
\usepackage{shuffle}
\usepackage{diagbox}
\usepackage{tabularx,booktabs}
\newcolumntype{Y}{>{\centering\arraybackslash}X}

\usepackage{algorithm}
\usepackage{listings}

\usepackage{graphicx}
\usepackage{algpseudocode}
\usepackage[capitalise]{cleveref}

\newcommand{\N}{\mathbb{N}}

\def\ta{\mathtt{a}}
\def\tb{\mathtt{b}}
\def\tc{\mathtt{c}}

\def\tn{\mathtt{n}}

\def\r{\operatorname{r}}

\def\m{\operatorname{m}}

\renewcommand{\epsilon}{\varepsilon}
\renewcommand{\phi}{\varphi}

\newcommand{\disjoincup}{\dot{\cup}}

\DeclareMathOperator{\ar}{ar}

\DeclareMathOperator{\al}{alph}
\DeclareMathOperator\ScatFact{ScatFact}

\DeclareMathOperator{\letters}{alph}
\DeclareMathOperator{\Fact}{Fact}

\DeclareMathOperator{\inner}{in}

\DeclareMathOperator{\cond}{cond}
\DeclareMathOperator{\SAS}{SAS}

\DeclareMathOperator{\sfe}{sf}

\DeclareMathOperator{\perfectshuffle}{\underline{\shuffle}}

\newcommand{\harpvecsign}{\scriptscriptstyle\leftharpoonup}
\newcommand{\harpoonvec}[2]{%
	\ifx\displaystyle#1\doalign{$\harpvecsign$}{#1#2}\fi
	\ifx\textstyle#1\doalign{$\harpvecsign$}{#1#2}\fi
	\ifx\scriptstyle#1\doalign{\scalebox{.6}[.9]{$\harpvecsign$}}{#1#2}\fi
	\ifx\scriptscriptstyle#1\doalign{\scalebox{.5}[.8]{$\harpvecsign$}}{#1#2}\fi
}
\newcommand{\doalign}[2]{%
	{\vbox{\offinterlineskip\ialign{\hfil##\hfil\cr#1\cr$#2$\cr}}}%
}

\def\nth#1{#1$^{\text{th}}$}

%% file: abstract.tex
\begin{abstract}
Starting in the 1970s with the fundamental work of Imre Simon, \emph{scattered factors} (also known as subsequences or scattered subwords) have remained a consistently and heavily studied object. 
The majority of work on scattered factors can be split into two broad classes of problems: given a word, what information, in the form of scattered factors, are contained, and which are not.
In this paper, we consider an intermediary problem, introducing the notion of \emph{complement scattered factors}. Given a word $w$ and a scattered factor $u$ of $w$, the complement scattered factors of $w$ with regards to $u$, $C(w, u)$, is the set of scattered factors in $w$ that can be formed by removing any embedding of $u$ from $w$. 
This is closely related to the \emph{shuffle} operation in which two words are intertwined, i.e., we extend previous work relating to the shuffle operator, using knowledge about scattered factors.
Alongside introducing these sets, we provide combinatorial results on the size of the set $C(w, u)$, an algorithm to compute the set $C(w, u)$ from $w$ and $u$ in $O(\vert w \vert \cdot \vert u \vert \binom{w}{u})$ time, where $\binom{w}{u}$ denotes the number of embeddings of $u$ into $w$, an algorithm to construct $u$ from $w$ and $C(w, u)$ in $O(\vert w \vert^2 \binom{\vert w \vert}{\vert w \vert - \vert u \vert})$ time, and an algorithm to construct $w$ from $u$ and $C(w, u)$ in $O(\vert u \vert \cdot \vert w \vert^{\vert u \vert + 1})$ time.
\end{abstract}

%% file: intro.tex
For a given word $w$, a \emph{scattered factor} (\emph{(scattered) subword or subsequence}) of $w$ 
is a word that is obtained by deleting letters from $w$ while preserving the order of the remaining letters, 
i.e., formally a word $u$ of length $n$ is a scattered factor of the word $w$ if there exist, possibly empty words, $v_1, \ldots, v_{n+1}$ such that $w = v_1 u[1] v_2 u[2] \cdots v_n u[n] v_{n+1}$ where $u[i]$ 
denotes the \nth{$i$} letter of $u$. For instance, $\mathtt{pile}, \mathtt{nap}$, and $\mathtt{pipe}$ are scattered factors of $\mathtt{pineapple}$ where as $\mathtt{peel}$ and $\mathtt{lane}$ are not as the letters do not 
occur in the correct order.  As scattered factors serve as a complexity measure for a word's absent or existing information, the relation between words and their scattered factors is studied in combinatorics on words, stringology, language and automata theory (see \cite[Chapter 6, Subwords]{lothaire} for a profound overview on the topic). The research on scattered factors
started with the PhD thesis of Imre Simon in the 1970s with the fundamental work about piecewise-testable languages 
\cite{DBLP:conf/automata/Simon75}. Simon introduced the relationship, nowadays known as Simon's congruence, in which two words are said to be congruent if they have the same scattered factors up to a given length $k$. 
As yet, the index of this congruence relation is unknown. An approach to get more insights about the index is the investigation of $m$-nearly $k$-universal words \cite{DBLP:conf/dlt/BarkerFHMN20,DBLP:journals/tcs/FleischmannHHHMN23,DBLP:conf/stacs/DayFKKMS21,fleischmann2021scattered,DBLP:conf/cwords/SchnoebelenV23}: a word is called \emph{$m$-nearly $k$-universal} if exactly $m$ words of length $k$ are not scattered factors of it. The notion of universality is tightly related to the notion of $k$-richness (cf. \cite{KarandikarKS15,CSLKarandikarS,journals/lmcs/KarandikarS19}).
Recently, Simon's congruence got attention in language theory \cite{DBLP:journals/tcs/KimHKS23,DBLP:conf/dlt/KimHKS23,DBLP:conf/isaac/AdamsonFHKMN23,DBLP:conf/isaac/FazekasKMMS24,DBLP:journals/corr/abs-2503-18611}
and pattern matching \cite{DBLP:conf/rp/FleischmannKKMNSW23,DBLP:journals/tcs/KimKH24}.
Moreover, scattered factors have been studied in combinatorics on words 
\cite{DBLP:journals/fuin/KoscheKMS22,DBLP:journals/jcss/MateescuSY04},
algorithms and stringology \cite{DBLP:journals/tcs/Baeza-Yates91,DBLP:journals/algorithmica/BannaiIKKP24,DBLP:conf/soda/BringmannK18,DBLP:journals/aam/ForresterM23,DBLP:journals/jacm/Maier78,DBLP:journals/jacm/WagnerF74},
logics \cite{DBLP:journals/lmcs/BaumannGTZ23,DBLP:conf/lics/HalfonSZ17,10.1007/978-3-030-50026-9_21,10.1007/978-3-030-17127-8_20},
as well as language and automata theory \cite{DBLP:conf/isaac/AdamsonFHKMN23,KarandikarKS15,DBLP:journals/lmcs/KarandikarS19,simon1972hierarchies,DBLP:conf/automata/Simon75,zetzsche:LIPIcs.ICALP.2016.123}. From a more applicative point of view, scattered factors are used in 
bioinformatics \cite{DBLP:journals/bib/DoLL21,DBLP:conf/aaai/WangCG20}, formal
software verification \cite{DBLP:conf/lics/HalfonSZ17,zetzsche:LIPIcs.ICALP.2016.123}
and database theory \cite{artikis2017complex,FrochauxK23,Kleest-Meissner22,Kleest-Meissner23,SchmidSIGMOD}.
Further, scattered factors are also used to model corrupted data in the context
of the theoretical problem of reconstruction
\cite{dress2005reconstructing,DBLP:journals/ijfcs/FleischmannLMNR21,DBLP:conf/dlt/Manuch99}.

As mentioned briefly regarding the notion of $m$-nearly $k$-universal words, one can not only investigate the 
scattered factors of a word but also characterise words by their absent scattered factors, e.g., \texttt{lane}
is an absent scattered factor of \texttt{pineapple} (cf. \cite{DBLP:journals/corr/abs-2407-18599,DBLP:journals/tcs/FleischmannHHHMN23,DBLP:conf/fct/FleischmannHHN23,DBLP:journals/fuin/KoscheKMS22,DBLP:conf/cocoa/Tronicek23}).
In bioinformatics not only absent scattered factors but also the remaining parts after deletion (or insertion, substitution) are considered that are created when DNA is passing a (deleting) DNA-storage channel \cite{sabary2024reconstruction,haghighat2025half}.
After a thorough search of the literature, it seems like nobody investigated the remaining words, i.e., the {\em complement set of scattered factors} for two given words $u$ and $w$
from a combinatorial perspective so far, i.e., the set of words which we obtain if we delete $u$ from $w$. For instance, \texttt{ala} has three
occurrences (embeddings) in \texttt{alfalfa} and if we delete them separately, we obtain $C(w,u)=\{\mathtt{flfa},\mathtt{falf},\mathtt{lfaf}\}$.  In this case, the cardinality of the complement set equates the number of embeddings of the deleted word, but  if we consider $\mathtt{peelwheel}$ and $\mathtt{peel}$ we get $|C(w,u)|=4$ even though there are
seven embeddings. This notion is tightly related to the notion of the \emph{shuffle operator}, but from the opposite perspective.
Given two words $u$ and $v$, the \emph{shuffle} $u\shuffle v$ is the set containing 
all words $w$ such that deleting $u$ from $w$ results in $v$ and deleting $v$ from $w$ results in 
$u$, i.e., $u\shuffle v=\{w\in\Sigma^{\ast}|\,\exists k,\ell\in\N\,\exists u_1,\ldots,u_k,v_1,\ldots ,v_{\ell}\in\Sigma^{\ast}:\, u=u_1\cdots u_k, v=v_1\cdots v_k, w=u_1v_1\cdots u_nv_n\}$.  For instance, given $\mathtt{lop}$ and $\mathtt{apt}$, we have 
$\mathtt{laptop}\in\mathtt{lop}\shuffle\mathtt{apt}$: deleting the first letter and the two last letters 
results in $\mathtt{apt}$ while deleting the second to the fourth letters results in $\mathtt{lop}$.
The shuffle has applications in programme verification, e.g. as a tool for modelling process algebras
\cite{baeten_weijland_1990}. The definition can be generalised to languages and the shuffle on regular languages
is investigated thoroughly \cite{sik1998ModelingLM,DBLP:journals/iandc/BerstelBCPR10}. Further results regarding the shuffle and formal languages can be found in \cite{DBLP:journals/ijfcs/CampeanuSV02,DBLP:journals/corr/abs-1208-2747}.
Within the research of the shuffle, two important subfields emerged: the perfect shuffle and the self shuffle (or shuffle square). For the \emph{perfect shuffle} the letters have to be taken alternating from each word, (cf. \cite{DIACONIS1983175,Ellis2002TheCO}).
Notice that in contrast to the general problem where the shuffle of two words is a set, the 
perfect shuffle of two words is just one word, e.g., the perfect shuffle of $\mathtt{bnn}$  and 
 $\mathtt{aaa}$ gives $\mathtt{banana}$. For the \emph{self shuffle} only the  shuffle of a word with itself, is investigated. Thus, given $w=\ta\tb\tc$, we get $w\shuffle 
w=\{\ta\tb\tc\ta\tb\tc,\ta\tb\ta\tc\tb\tc,\ta\tb\ta\tb\tc\tc,\ta\ta\tb\tc\tb\tc,\ta\ta\tb\tb\tc\tc\}
$. 
The notion of the self shuffle has been investigated lately in several 
papers (see, e.g. \cite{DBLP:conf/cwords/Harju13,DBLP:journals/tcs/HarjuM15a,Currie2014squarefree}). Here, the shuffle of 
square free words, i.e., words that have no factor of the form $vv$ for some non-empty word $v$, 
was studied.
The decision problem whether a word is a self shuffle has been independently shown to be 
$\mathtt{NP}$-complete by \cite{Buss2014UnshufflingAS} and \cite{10.1007/978-3-642-38536-0_21}, 
although the alphabets were restricted. Recently, \cite{Bulteau2020RecognizingBS} showed the result 
for binary alphabets in general. For a collection of the recent developments and open problems regarding self shuffle, we refer to \cite{10.1007/978-3-319-15579-1_5}. Thus, coming back to the complement scattered factors, we are interested in the set $C(w,u)=\{v\in\Sigma^{\ast}|\,w\in u\shuffle v\}$ and its properties.

\textbf{Our Contribution.}
In this work, we introduce the notion of complement scattered factors and start their investigation from both  algorithmic and combinatorial perspectives. First, we compute the set $C(w,u)$ for given $w$ and $u$ based on a dynamic programming approach. 
Further, we give algorithms on two inverse problems: (1) Given a word $u$ and a set $S$, decide whether there exists a word $w$ such that $S = C(w,u)$ and if so, compute $w$, and (2) Given a word $w$ and a set $S$, decide whether there exists a scattered factor $u$ of $w$ such that $S = C(w,u)$ and if so, compute $u$. An overview of these three related problems in visualised in \Cref{fig:triangle} where the runtime refers to the algorithm where the two values in the resp.~adjacent edge are the input and the value on the opposite side the output. Moreover, we investigate complement scattered factors from a combinatorics of words perspective: First, we present special scattered factors that have exactly one or multiple complement scattered factors by using tools as the arch or the $\alpha$-$\beta$-factorisation. In a next step, we investigate the set of words for which $C(w,u) =1$ iff $E(w,u) = 1$ holds where $E(w,u)$ denotes the number of embeddings of $u$ in $w$. As a concluding step, we relate the notion of complement scattered factors with the (self) shuffle operation and give a characterisation of complement scattered factors within words that are obtained by the perfect shuffle operation.

\textbf{Structure of the Work.}
In \Cref{sec:prelims} we present all the necessary definitions. Afterwards in \Cref{algo} we present the algorithmical investigation of computing $C(w,u)$. Afterwards, in \Cref{sec:csf} we investigate the object of complement scattered factors from a combinatorics on words perspective. 

\begin{figure}
    \centering
    \begin{tikzpicture}[scale=0.75]

    \coordinate (A) at (0,0);
    \coordinate (B) at (4,0);
    \coordinate (C) at (2,3.464);

    \node (S) at (2,2.7) {$S$};
    \node (eq) at (2,2.7-0.4) {$=$};
    \node (Cwu) at (2,2.7-0.4-0.3) {$C(w,u)$};
    \node (length) at (2,2.7-0.4-0.3-0.4) {$\subseteq \Sigma^\ell$};

    \node (w) at (0.7,0.4) {$w$};
    \node (u) at (4-0.7,0.4) {$u$};

    \coordinate[label=below:{$O(|w||u|\binom{w}{u})$}](c) at ($ (A)!.5!(B) $);
    \node (alg1) at ($(A)!.5!(B) -(0,0.9)$) {\Cref{prop:runtime_complement_set}};

    \coordinate (b) at ($ (A)!.5!(C) $);
    \node[rotate=60] (alg2) at ($(A)!.5!(C) -(0.5,0)$) {$O(|w|^2\binom{|w|}{\ell})$};
    \node[rotate=60] (alg2) at ($(A)!.5!(C) -(0.5,0)-(0.6,-0.2)$) {\Cref{Swuproblem}};

    \coordinate(a) at ($ (B)!.5!(C) $);
    \node[rotate=-60] (alg3) at ($(B)!.5!(C) +(0.3,0.3)$) {$O(|S| (\ell-|u|)^{|S| + 1})$};
    \node[rotate=-60] (alg3) at ($(B)!.5!(C) +(0.3,0.3)+(0.4,0.4)$) {\Cref{thm:compliment_to_w}};

    \draw [line width=1.5pt] (A) -- (B) -- (C) -- cycle;
    \end{tikzpicture}
    \caption{A summary of the algorithmic results of \Cref{prop:runtime_complement_set,Swuproblem,thm:compliment_to_w}.}
    \label{fig:triangle}
\end{figure}

%% file: prelims.tex
Let $\mathbb{N}$ be the set of natural numbers, $\mathbb{N}_0 = \mathbb{N} \cup \{0\}$, $[n] = \{1,\ldots,n\}$, $[n]_0 := [n] \cup \{0\}$. 
Let $M,N$ be sets and $(f_i)_{i\in\N}$ a family of functions from $M$ to $N$. Set $f_i(M)=\{f_i(m)|\,m\in M\}$ for all $i\in\N$.
For $i,j\in\N$ distinct, $f_i$ and $f_j$ are called {\em disjoint} if $f_i(M)\cap f_j(M)=\emptyset$ and {\em exhaustive} if $f_i(M)\cup f_j(M)=N$. If two functions $f,g: M\rightarrow N$ are disjoint and exhaustive, we denote that by $f\disjoincup g$.

An \emph{alphabet} $\Sigma$ is a non-empty, finite set whose elements are called \emph{letters} or {\em symbols}.
For a given $\sigma\in\N$, we fix  $\Sigma=\{\ta_1,\ldots,\ta_{\sigma}\}$.  A \emph{word} is a finite sequence of letters from $\Sigma$. Let $\Sigma^*$ be 
the set of all finite words over $\Sigma$ with concatenation and the empty word $\epsilon$ as 
neutral element.
Set $\Sigma^+ := \Sigma^* \setminus \{\epsilon\}$. Let $w\in\Sigma^{\ast}$. 
For all $n\in\N_0$ define inductively, $w^0=\varepsilon$ and $w^n=ww^{n-1}$. 
The {\em length} of $w$ is the number of $w$'s letters; thus $|\epsilon| = 0$. 
For all $k \in \mathbb{N}_0$ set $\Sigma^k := \{w \in \Sigma^* \mid |w| = k\}$ (define $\Sigma^{\leq k}$, $\Sigma^{\geq k}$ analogously). We denote $w$'s \nth{$i$} letter by $w[i]$ 
and set $w[i..j]=w[i]\cdots w[j]$ if $i < j$,  and $\epsilon$ if $ i > j$ for all $i,j \in [\vert w \vert]$. 
Set $\letters(w) = \{\ta \in \Sigma \mid \exists i \in [|w|]: w[i] = \ta \}$. The word $u \in 
\Sigma^*$ is a \emph{factor} of $w$ if there exist $x,y \in \Sigma^*$ such that $w = xuy$.
In the case $x=\epsilon$, we call $u$ a \emph{prefix} of $w$ and \emph{suffix} if $y = \epsilon$.
We call a factor $\ta\ta$, for some $\ta\in\Sigma$, of a word $w\in\Sigma^{\ast}$ a {\em letter square of $w$} and a factor $yy \in \Fact(w)$ for some $y \in \Sigma^*$ a \emph{square} in $w$.
Define the \emph{reverse} of $w$ by $w^R = w[|w|]\cdots w[1]$.
Let $<_{\Sigma}$ be a total order on $\Sigma$ and denote the \emph{lexicographical 
	order} on $\Sigma^*$ by $<$.
For further definitions see \cite{lothaire}. After these basic notations, we introduce the scattered factors.

\begin{definition}
	Let $w \in \Sigma^*$ and $n \in \mathbb{N}_0$. A word $u \in \Sigma^n$ is called a \emph{scattered factor} of w - denoted as $u \in \ScatFact(w)$ - if there exist $v_1,\dots, v_{n+1} \in \Sigma^*$ such that $w = v_1 u[1] v_2 u[2] \cdots v_n u[n] v_{n+1}$.
	Set $\ScatFact_k(w)=\{u\in\ScatFact(w)|\,|u|=k\}$. 
\end{definition}

Now we define the set of positions a scattered factor occurs in, given by mappings that assign the indices of a scattered factor $u$ to the resp.~ones in $w$.

\begin{definition}
	An \emph{embedding of $u \in \Sigma^*$ in $w\in \Sigma^*$} is defined as the mapping $e: \{1, \ldots, |u|\} \to \{1, \ldots, |w|\}$ with $e(1) < \cdots < e(|u|)$ s.t. $u[i] = w[e(i)]$ for all $i \in [|u|]$.
	With $\binom{w}{u}$ we denote the number of embeddings of $u$ in $w$, i.e., $\binom{w}{u} = |\{e : \{1, \ldots, |u|\} \to \{1, \ldots, |w|\} \mid e \text{ is an embedding of $u$ in $w$}\}|$.
	For a given embedding $e$, let $\sfe(e)$ be the scattered factor defined by $e$. The \emph{set of all embeddings of $u$ in $w$} will be denoted by $E(w,u)$. For $e\in E(w,u)$ let $\overline{e}$
	denote the embedding of a word $v$ such that $e\dot{\cup}\overline{e}$.
	Denote by $E_=(w,u)$ the embeddings of $u$ in $w$ for which we have $\sfe(e_1)=\sfe(e_2)$ for all $e_1,e_2\in E_=(w,u)$. 
\end{definition}

The words $\mathtt{cau}, \mathtt{cafe}$, $\mathtt{life}$ and $\mathtt{ufo}$ are all scattered factors of $\mathtt{cauliflower}$ but neither are $\mathtt{flour}$ nor $\mathtt{row}$. Based on the tuple notation for permutations, we use a $|u|$-tuple for denoting an embedding of a scattered factor $u$ in a word $w$ such that the \nth{$i$} entry of the tuple contains the index in $w$ where the \nth{$i$} letter of $u$ is embedded. For example, the only embedding of $\mathtt{cafe}$ in $\mathtt{cauliflower}$ is given by $(1,2,6,10)$. Notice that scattered factors might have multiple embeddings like $\mathtt{low}$ given by $(4,8,9)$ and $(7,8,9)$. In \texttt{bananaban} there are multiple occurrences of $\mathtt{ban}$ some disjoint like $(1,2,3)$, $(7,8,9)$ and some not like $(1,2,9)$, $(1,4,5)$. The words $\mathtt{banana}$ and $\mathtt{ban}$ have a disjoint and exhaustive embedding in $\mathtt{bananaban}$. We continue with the concept of (nearly) universality introduced in \cite{DBLP:conf/dlt/BarkerFHMN20,DBLP:journals/tcs/FleischmannHHHMN23}.

\begin{definition}
	A word $w \in \Sigma^*$ is called \emph{$m$-nearly $k$-universal} if $\vert \ScatFact_k(w) \vert $$= 
	\sigma^k-m$. 
	If $m=0$, the words are called $k$-universal.
	The largest $k\in\N_0$, such that a word is $k$-universal, is called the {\em universality index} $\iota(w)$.
\end{definition}

In the unary alphabet all words are of the form $\ta^{\ell}$ for $\ell\in\N$ and the scattered factors can be determined easily.  Therefore, we only consider alphabets of size two or greater. Moreover, we assume $\Sigma=\letters(w)$ for a given $w$, if not stated otherwise.
One of the main tools for the investigation of a words' scattered factors is the \emph{arch factorisation} which was introduced by Hébrard \cite{DBLP:journals/tcs/Hebrard91}. 

\begin{definition}
	For a word $w \in \Sigma^*$ the \emph{arch factorisation} is given by $w = \ar_1(w) \cdots \ar_k(w) \r(w)$ for $k \in \mathbb{N}_0$ with\\
	(1) $\iota(\ar_i(w))=1$ for all $i \in [k]$, \\
	(2) $\ar_i(w)[\vert \ar_i(w) \vert] \notin \letters(\ar_i(w)[1 .. \vert \ar_i(w) \vert - 1 ])$ for all $i \in [k]$, and\\
	(3) $\letters(\r(w)) \subsetneq \Sigma$.\\
	The words $\ar_i(w)$ are called \emph{arches} of $w$ and $\r(w)$ is the \emph{rest} of $w$.
	The {\em modus} $\m(w)$ is given by  $\ar_1(w)[\vert\ar_1(w)\vert] \cdots\ar_k(w)[\vert\ar_k(w)\vert]$, i.e., it is the word containing the unique last letters of each arch.
	The \emph{inner of the $i^{th}$ arch of $w$} is defined as the prefix of 
	$\ar_i(w)$ such that $\ar_i(w) = \inner_i(w) 
	\m(w)[i]$ holds.
\end{definition}

Tightly related to the notion of scattered factors is the {\em shuffle} (cf. \cite{JSTOR:AoM/1969820}).

\begin{definition}\label{def:shuffle_occs}
  For $u,v\in\Sigma^{\ast}$ define the {\em shuffle} of $u$ and $v$ by
  $u\shuffle v=\{w\in\Sigma^{|u|+|v|}|\,\exists e_1\in E(w,u)\,\exists e_2\in E(w,v):\, e_1\disjoincup e_2\}$.
  For $u,v\in\Sigma^{m}$ for $m \in \N$ define the {\em perfect shuffle of $u$ and $v$} as $u\perfectshuffle v=u[1]v[1]u[2]v[2]\cdots u[|u|]v[|u|]$. The set $u \shuffle u$ is called the \emph{selfshuffle} of $u$.
\end{definition}

For instance, let $u=\mathtt{ban}$ and $v=\mathtt{ana}$. Thus, the shuffle contains words of length $6$ and we have $|u\shuffle v|=11$. In general, we have $|u\shuffle v|\leq\binom{|uv|}{|u|}$.

\begin{definition}
	For $w \in \Sigma^*$ and $u\in \ScatFact(w)$ we define the complement scattered factor set of $u$ in $w$ by $C(w,u) = \{v\in\Sigma^{|w|-|u|}|\,w\in u\shuffle v\}$. Each $v\in C(w,u)$ is called a {\em complement scattered factor} of $u$ in $w$.
\end{definition}

Observe that, by definition, if $v \in C(w, u)$ then $u \in C(w, v)$.
We finish this section with a lemma on the selfshuffle and its embeddings (cf. \cite[Mathematical Preliminaries]{Buss2014UnshufflingAS}).

\begin{lemma}\label{lem:fst_snd_occurence}
  For $v\in\Sigma^{\ast}$ and \(w\in v\shuffle v\) there exist $e_1,e_2\in E(w,v)$ with $e_1\dot{\cup}e_2$  such that $e_1(k) < e_2(k)$ for all 
\(k\in [|v|]\).
  We call \(e_1\) the \emph{first} and \(e_2\) the \emph{second} occurrence of \(v\) in \(w\) if for all $e_3\in E(w,v)$
  there exist $j\in[|v|]$ s.t.~$e_1(j)\leq e_3(j)$.
\end{lemma}

%
%
%
%

For the algorithmic results we use the standard 
computational model RAM with logarithmic word-size (see, e.g., \cite{KarkkainenSB06}), i.e., we 
follow a standard assumption from stringology, that if $w$ is the 
input word for our algorithms, then we assume $\Sigma=\letters(w)=\{ 1,2, \ldots, \sigma\}$.

%% file: computingC.tex
We first present \Cref{alg:computing_C} for computing the complement scattered factor set from a given pair of words $w \in \Sigma^*$ and $u \in \ScatFact(w)$ via a dynamic programming approach.
We maintain a dynamic programming matrix $P$ which contains in $P[i,j]$ the complement scattered factor set
of $u[1..i-1]$ in $w[1..j]$.

\begin{algorithm}
	\textbf{Input:} $w \in \Sigma^*$ and $u \in \ScatFact(w)$\\
	\textbf{Output:} The set $C(w,u)$.
	\begin{lstlisting}[mathescape=true]
Create an $|u|+1 \times |w|$-table $P$ initially filled with $\emptyset$
// initialisation of the first column
$P[1,1] = \{w[1]\}$
if $w[1] = u[1]$ then
  $P[2,1] = \{\varepsilon\}$
for $j \in [2,|w|]$:
  for $i \in [2,|u|+1]$:
    $P[1,j] = \{w[1..j]\}$
    if $i -1> j$ then 
      // here the prefix of $u$ is longer than the prefix of $w$
      $P[i,j] = \emptyset$
    else
      if $w[j] = u[i-1]$ then
        $P[i,j] = P[i-1,j-1] \cup P[i,j-1] \cdot w[j]$ 
        // the second term of the union is empty if $P[i,j-1] = \emptyset$
      else 
        $P[i,j] = P[i,j-1] \cdot w[j]$ 
return $P[|u|+1,|w|]$
	\end{lstlisting}
	\caption{Compute $C(w,u)$}\label{alg:computing_C}
\end{algorithm}

\ifpaper
\else
\input{correctness_of_P}
\input{proof_correctness_of_P}
\fi

From \Cref{prop:correctness_of_P}, we immediately get the following two results on the length of every  entry $P[i,j]$ and the cases where only the empty set is stored (i.e., the prefix of $u$ is not a scattered factor of the prefix of $w$).

\begin{corollary}
	For $w \in \Sigma^*$, $u \in \ScatFact(w)$, $i \in [|u|+1]$ and $j \in [|w|]$, we have $|v| = |w[1..j]|-|u[1..i-1]|$ for all $v \in P[i,j]$.
\end{corollary}

\begin{corollary}
	For $w \in \Sigma^*$, $u \in \ScatFact(w)$, $i-1 > j$ we have $P[i,j] = \emptyset$.
\end{corollary}

Thus, we get the complement scattered factor set of two given words $w \in \Sigma^*$ and $u \in \ScatFact(w)$ by computing the table $P$.

\ifpaper
\else
\input{runtime_complement_set}
\input{proof_runtime_complement_set}
\fi

\ifpaper
\else
\Cref{example:algoComplementSet} provides an example of \Cref{alg:computing_C} on the input $C(\mathtt{ananas}, \mathtt{as})$.
\begin{table}
	\resizebox{\textwidth}{!}{
		\begin{tabular}{c | c c c  c c c c c c c c}
			& $w[1..6]$ && $w[2..6] $ && $w[3..6] $ && $w[4..6] $ && $w[5..6] $ && $w[6]$\\
			& $\mathtt{\underline{a}}$ & & $\mathtt{a\underline{n}}$ & & $\mathtt{an\underline{a}}$ & & $\mathtt{ana\underline{n}}$ & & $\mathtt{anan\underline{a}}$ && $\mathtt{anana\underline{s}}$\\
			\hline
			$\varepsilon$ & $\{\mathtt{a}\}$ & & $\{\mathtt{an}\}$ & & $\{\mathtt{ana}\}$  & & $\{\mathtt{anan}\}$ & & $\{\mathtt{anana}\}$ && $\{\mathtt{ananas}\}$\\
			& & & &$\searrow$& & & &$\searrow$& &&\\
			$u[1] = \mathtt{\underline{a}}$ & $\{\varepsilon\}$ & $\overset{\cdot \tn}{\rightarrow}$ & $\{\mathtt{n}\}$ & $\overset{\cdot \ta}{\rightarrow}$ & $\{\mathtt{an}, \mathtt{na}\}$ & $\overset{\cdot \tn}{\rightarrow}$ & $\{\mathtt{ann}, \mathtt{nan}\}$ & $\overset{\cdot \ta}{\rightarrow}$ & $\{\mathtt{anna}, \mathtt{nana}$ &$\overset{\cdot \mathtt{s}}{\rightarrow}$& $\{\mathtt{annas}, \mathtt{nanas},$\\
			&  && && && && $\mathtt{anan}\}$ && $\mathtt{anans}\}$\\
			& && && && && &$\searrow$& \\
			$u[1..2] = \mathtt{a\underline{s}}$ & $\emptyset$ &$\overset{\cdot\ta}{\rightarrow}$& $\emptyset$ &$\overset{\cdot\tn}{\rightarrow}$& $\emptyset$ & $\overset{\cdot\ta}{\rightarrow}$ & $\emptyset$ & $\overset{\cdot\tn}{\rightarrow}$ & $\emptyset$ &$\overset{\cdot\mathtt{s}}{\rightarrow}$ & $\{\mathtt{anna}, \mathtt{nana}$\\
			& && && && && && $\mathtt{anan}\}$\\
	\end{tabular}}
	\caption{Computation of $C(\mathtt{ananas}, \mathtt{as})$.}\label{example:algoComplementSet}
\end{table}
\fi
\ifpaper
\else 
Further, \Cref{alg:computing_C+counting} is an extended version of \Cref{alg:computing_C} for not only computing the set of complement scattered factors but also their multiplicities (this algorithm does not work via matching the prefixes but the suffixes of $w$ and $u$).

\begin{algorithm}
\textbf{Input:} $w \in \Sigma^*$ and $u \in \ScatFact(w)$\\
\textbf{Output:} The set $C(w,u)$ equipped with the multiplicity of occurrences of every element.
\begin{lstlisting}[mathescape=true]
Create an $|u|+1 \times |w|$-table $P$ initially filled with $\emptyset$
// initialisation of the last column
 $P[1,|w|] = \{(w[|w|],1)\}$
if $w[|w|] = u[|u|]$ then
  $P[2,|w|] = \{(\varepsilon,1)\}$
for $j \in [|w|-1,1]$:
  $P[1,j] = \{(w[j..|w|],1)\}$
  for $i \in [2,|u|+1]$:
    if $i -1> |w|-j+1$ then 
      // here the suffix of $u$ is to longer than the suffix of $w$ 
      // (which has length $|w|-j+1$)
      $P[i,j] = \emptyset$
    else
      if $w[j] = u[i-1]$ then
        $P[i,j] = $unionWithMultiplicity$(P[i-1,j+1],P[i,j+1],w[j])$ 
        // the second term of the union is empty if $P[i,j-1] = \emptyset$
      else 
        $P[i,j] = \{(w,n) \mid w \in P[i,j+1].\mathtt{first}() \cdot w[j] \land P[i,j+1].\mathtt{snd}()\}$ 
return $P[|u|+1,|w|]$

// type $\mathcal P(\Sigma^* \times \N) \times \mathcal P(\Sigma^* \times \N) \times \Sigma \to \mathcal P(\Sigma^* \times \N)$
def unionWithMultiplicity(compSet1, compSet2,letter):
// using a union find datastructure
    compSet$ = \emptyset$
    for all $w \in$letter$\cdot$compSet1.first()$\cup$compSet2.first():
      if $w \in$letter$\cdot$compSet1.first()$\cap$compSet2.first():
    	compSet$=$compSet$\cup (w,$compSet1.snd()+compSet2.snd()$)$
      else if $w \in$letter$\cdot$compSet1.first():
        compSet$=$compSet$\cup (w,$compSet1.snd()$)$
      else:
        compSet$=$compSet$\cup (w,$compSet2.snd()$)$
    return compSet
\end{lstlisting}
\caption{Compute $C(w,u)$ via suffix matching and counting}\label{alg:computing_C+counting}
\end{algorithm}

\begin{example}
	We successively construct the complement scattered factors for $w = \mathtt{ababa}$ and $u = \mathtt{aba}$ (the description of the table is just added for better readability, there the underlined lettered are those of interest to the algorithm). First, the last column (first three lines of the code):
	
	\begin{tabular}{c | c c c  c c c c c c}
		& $w[1..5]$ && $w[2..5] $ && $w[3..5] $ && $w[4..5] $ && $w[5] $\\
		& $\mathtt{\underline{a}baba}$ & & $\mathtt{\underline{b}aba}$ & & $\mathtt{\underline{a}ba}$ & & $\mathtt{\underline{b}a}$ & & $\mathtt{\underline{a}}$\\
		\hline
		$\varepsilon$ & $\emptyset$ & & $\emptyset$ & & $\emptyset$ & & $\emptyset$ & & $\{(\mathtt{a},1)\}$ \\
		$u[3] = \mathtt{\underline{a}}$ &  $\emptyset$ & & $\emptyset$ & & $\emptyset$ & & $\emptyset$ & & $\{(\varepsilon,1)\}$\\
		$u[2..3] = \mathtt{\underline{b}a}$ & $\emptyset$ & & $\emptyset$ & & $\emptyset$ & & $\emptyset$ & & $\emptyset$\\
		$u[1..3] = \mathtt{\underline{a}ba}$ & $\emptyset$ & & $\emptyset$ & & $\emptyset$ & & $\emptyset$ & & $\emptyset$ 
	\end{tabular}
	
	Now, we iterate through $w$ and then $u$, i.e., full the table column-wise from right to left.
	
	\begin{tabular}{c | c c c  c c c c c c}
		& $w[1..5]$ && $w[2..5] $ && $w[3..5] $ && $w[4..5] $ && $w[5] $\\
		& $\mathtt{\underline{a}baba}$ & & $\mathtt{\underline{b}aba}$ & & $\mathtt{\underline{a}ba}$ & & $\mathtt{\underline{b}a}$ & & $\mathtt{\underline{a}}$\\
		\hline
		$\varepsilon$ & $\emptyset$ & & $\emptyset$ & & $\emptyset$ & & $\{(\mathtt{ba},1)\}$ & & $\{(\mathtt{a},1)\}$ \\
		$u[3] = \mathtt{\underline{a}}$ &  $\emptyset$ & & $\emptyset$ & & $\emptyset$ & & $\{(\tb,1)\}$ & $\overset{\tb \cdot}{\leftarrow}$ & $\{(\varepsilon,1)\}$\\
		& && && && &$\swarrow$& \\
		$u[2..3] = \mathtt{\underline{b}a}$ & $\emptyset$ & & $\emptyset$ & & $\emptyset$ & & $\{\varepsilon,1\}$ & $\overset{\tb \cdot}{\leftarrow}$ & $\emptyset$\\
		$u[1..3] = \mathtt{\underline{a}ba}$ & $\emptyset$ & & $\emptyset$ & & $\emptyset$ & & $\emptyset$ & $\overset{\tb \cdot}{\leftarrow}$ & $\emptyset$ 
	\end{tabular}
	
	Recall that we only use diagonal values if the resp. letters of $u$ and $w$ match.
	
	\begin{tabular}{c | c c c  c c c c c c}
		& $w[1..5]$ && $w[2..5] $ && $w[3..5] $ && $w[4..5] $ && $w[5] $\\
		& $\mathtt{\underline{a}baba}$ & & $\mathtt{\underline{b}aba}$ & & $\mathtt{\underline{a}ba}$ & & $\mathtt{\underline{b}a}$ & & $\mathtt{\underline{a}}$\\
		\hline
		$\varepsilon$ & $\emptyset$ & & $\emptyset$ & & $\{(\mathtt{aba},1)\}$  & & $\{(\mathtt{ba},1)\}$ & & $\{(\mathtt{a},1)\}$ \\
		& && && &$\swarrow$& && \\
		$u[3] = \mathtt{\underline{a}}$ &  $\emptyset$ & & $\emptyset$ & & $\{(\mathtt{ba},1),$ & $\overset{\ta \cdot}{\leftarrow}$ & $\{(\tb,1)\}$ & $\overset{\tb \cdot}{\leftarrow}$ & $\{(\varepsilon,1)\}$\\
		& && && $(\mathtt{ab},1)\}$ && &&\\
		& && && && &$\swarrow$& \\
		$u[2..3] = \mathtt{\underline{b}a}$ & $\emptyset$ & & $\emptyset$ & & $\{\ta,1\}$ & $\overset{\ta \cdot}{\leftarrow}$ & $\{\varepsilon,1\}$ & $\overset{\tb \cdot}{\leftarrow}$ & $\emptyset$\\
		& && && &$\swarrow$& && \\
		$u[1..3] = \mathtt{\underline{a}ba}$ & $\emptyset$ & & $\emptyset$ & & $\{\varepsilon,1\}$ &  $\overset{\ta \cdot}{\leftarrow}$ & $\emptyset$ & $\overset{\tb \cdot}{\leftarrow}$ & $\emptyset$ 
	\end{tabular}
	
	\medskip
	
	\begin{tabular}{c | c c c  c c c c c c}
		& $w[1..5]$ && $w[2..5] $ && $w[3..5] $ && $w[4..5] $ && $w[5] $\\
		& $\mathtt{\underline{a}baba}$ & & $\mathtt{\underline{b}aba}$ & & $\mathtt{\underline{a}ba}$ & & $\mathtt{\underline{b}a}$ & & $\mathtt{\underline{a}}$\\
		\hline
		$\varepsilon$ & $\emptyset$ & & $\{(\mathtt{baba},1)\}$ & & $\{(\mathtt{aba},1)\}$  & & $\{(\mathtt{ba},1)\}$ & & $\{(\mathtt{a},1)\}$ \\
		& && && &$\swarrow$& && \\
		$u[3] = \mathtt{\underline{a}}$ &  $\emptyset$ & & $\{(\mathtt{bba},1),$ & $\overset{\tb \cdot}{\leftarrow}$ & $\{(\mathtt{ba},1),$ & $\overset{\ta \cdot}{\leftarrow}$ & $\{(\tb,1)\}$ & $\overset{\tb \cdot}{\leftarrow}$ & $\{(\varepsilon,1)\}$\\
		& && $(\mathtt{bab},1)\}$ && $(\mathtt{ab},1)\}$ && &&\\
		& && &$\swarrow$& && &$\swarrow$& \\
		$u[2..3] = \mathtt{\underline{b}a}$ & $\emptyset$ & & $\{(\mathtt{ba},2),$ &$\overset{\tb \cdot}{\leftarrow}$& $\{(\ta,1)\}$ & $\overset{\ta \cdot}{\leftarrow}$ & $\{(\varepsilon,1)\}$ & $\overset{\tb \cdot}{\leftarrow}$ & $\emptyset$\\
		& && $(\mathtt{ab},1)\}$ && && &&\\
		& && && &$\swarrow$& && \\
		$u[1..3] = \mathtt{\underline{a}ba}$ & $\emptyset$ & & $\{(\mathtt{b},1)\}$ &$\overset{\tb \cdot}{\leftarrow}$& $\{(\varepsilon,1)\}$ &  $\overset{\ta \cdot}{\leftarrow}$ & $\emptyset$ & $\overset{\tb \cdot}{\leftarrow}$ & $\emptyset$ 
	\end{tabular}
	
	\medskip
	
	\begin{tabular}{c | c c c  c c c c c c}
		& $w[1..5]$ && $w[2..5] $ && $w[3..5] $ && $w[4..5] $ && $w[5] $\\
		& $\mathtt{\underline{a}baba}$ & & $\mathtt{\underline{b}aba}$ & & $\mathtt{\underline{a}ba}$ & & $\mathtt{\underline{b}a}$ & & $\mathtt{\underline{a}}$\\
		\hline
		$\varepsilon$ & $\{(\mathtt{ababa},1)\}$ & & $\{(\mathtt{baba},1)\}$ & & $\{(\mathtt{aba},1)\}$  & & $\{(\mathtt{ba},1)\}$ & & $\{(\mathtt{a},1)\}$ \\
		& &$\swarrow$& && &$\swarrow$& && \\
		$u[3] = \mathtt{\underline{a}}$ & $\{(\mathtt{baba},1),$ &$\overset{\ta \cdot}{\leftarrow}$& $\{(\mathtt{bba},1),$ & $\overset{\tb \cdot}{\leftarrow}$ & $\{(\mathtt{ba},1),$ & $\overset{\ta \cdot}{\leftarrow}$ & $\{(\tb,1)\}$ & $\overset{\tb \cdot}{\leftarrow}$ & $\{(\varepsilon,1)\}$\\
		& $(\mathtt{abba},1),$ && $(\mathtt{bab},1)\}$ && $(\mathtt{ab},1)\}$ && &&\\
		& $(\mathtt{abab},1)\}$ && && && &&\\
		& && &$\swarrow$& && &$\swarrow$& \\
		$u[2..3] = \mathtt{\underline{b}a}$ & $\{(\mathtt{aba},2),$ &$\overset{\ta \cdot}{\leftarrow}$& $\{(\mathtt{ba},2),$ &$\overset{\tb \cdot}{\leftarrow}$& $\{(\ta,1)\}$ & $\overset{\ta \cdot}{\leftarrow}$ & $\{(\varepsilon,1)\}$ & $\overset{\tb \cdot}{\leftarrow}$ & $\emptyset$\\
		& $(\mathtt{aab},1)\}$ && $(\mathtt{ab},1)\}$ && && &&\\
		& &$\swarrow$& && &$\swarrow$& && \\
		$u[1..3] = \mathtt{\underline{a}ba}$ & $\{(\mathtt{ba},2),$ &$\overset{\ta \cdot}{\leftarrow}$& $\{(\mathtt{b},1)\}$ &$\overset{\tb \cdot}{\leftarrow}$& $\{(\varepsilon,1)\}$ &  $\overset{\ta \cdot}{\leftarrow}$ & $\emptyset$ & $\overset{\tb \cdot}{\leftarrow}$ & $\emptyset$\\
		& $(\mathtt{ab},2)\}$ && && && &&\\
	\end{tabular}
	
	Now, the complement scattered factor set is $C(\mathtt{ababa},\mathtt{aba}) = \{\mathtt{ab},\mathtt{ba}\}$ where both of the complement scattered factors occur with a multiplicity of two.
\end{example}
\fi

We consider two related problems of computing $C(w,u)$ for given $w$ and $u$:
\begin{description}
\item[Problem 1.] Given $w\in\Sigma^{n}$ and $S\subseteq \ScatFact_{m}(w)$ for some $n,m \in \N$ determine whether there exists $u\in\ScatFact_{n-m}(w)$ such that $C(w,u)=S$.\label{problem1}
\item[Problem 2.] Given $u\in\Sigma^{\ast}$ and $S\subseteq\Sigma^{\ell}$ for some fixed $\ell\in\N$ determine whether there exists $w\in\Sigma^{|u|+\ell}$ such that $C(w,u)=S$.\label{problem2}
\end{description}

Let us now assume that we are given some word $w \in \Sigma^n$, and the set $S$ of scattered factors of $w$, all of length $m \in \N$,  after some unknown word $u$ has been removed (Problem 1). We show that $u$ can be determined in $O(n^2\binom{n}{m})$. Notice that by the definition of complement scattered factors, all $s\in S$ have to have the same length.
Our algorithm works by the following lemma.

\ifpaper
\else
\input{intersection}
\input{proofintersection}
\fi

From Lemma \ref{lem:u_from_w_and_s_u_in_intersection}, we can determine the value of $u$ from $\bigcap_{v \in S} C(w, v)$. We do so in a two step manner. First, we determine the candidate set $\mathcal{U} = \bigcap_{v \in S} C(w, v)$ of potential values of $u$. Next, we check each word in $u \in \mathcal{U}$, determining if $C(w, u) = S$. If any $u$ for which $C(w, u) = S$ is found, then the algorithm halts returning this word as the solution. Otherwise, the algorithm continues until the set has been exhausted. This gives the following theorem.

\ifpaper
\else
\input{Swuproblem}
\input{proofSwuproblem}
\fi


We now move on to Problem 2 of determining whether there exists some $w$ such that $C(w, u) = S$ for some given word $u$ and set of $\ell$-length words $S$. Our focus is on the more general problem of asking if, for a set of tuples $Z = \{(v_1, u_1), (v_2, u_2), \dots, (v_r, u_r)\}$ where $\vert v_i u_i\vert = \vert v_j u_j \vert$, $\forall i, j \in [r]$, $r\in\N$, there exists some word $w$ of length $\vert v_i u_i\vert$ such that there exists for each tuple $(v_i, u_i)$ a disjoint pair of embeddings of $v_i$ and $u_i$ into $w$.  We call the problem of finding such a word the \emph{pairwise disjoint embedding problem}. Note that this problem becomes equivalent to the original problem when the input set of tuples is $\{(v_1, u), (v_2, u), \dots, (v_r, u)\}$. Our primary result for this section is an $O(r n^r)$ time algorithm for solving the pairwise disjoint embedding problem for an input set $Z$ of $r$ tuples each with length $n$, i.e., for the tuple $(v_i, u_i) \in Z$, $\vert v_i u_i \vert = n$.

We solve the pairwise disjoint embedding problem via a dynamic programming approach, constructing a table $V$ (for \textbf{V}alid) of size $O(n^{\vert Z \vert + 1})$ which we populate in $O(\vert Z \vert)$ time per cell. Formally, each entry in table $V$ is indexed by a set of $\vert Z \vert$ tuples $\{(z_1, y_1), (z_2, y_2), \dots, (z_{r}, y_{r})\}$ where $z_i$ is a suffix of $v_i$, and $y_i$ is a suffix of $u_i$. We further assume that, for each such tuple, $\vert z_i y_i \vert = \vert z_j y_j \vert$, for all $i, j \in [r]$. Note that this limitation does not apply between sets, thus we may have a pair of such sets $(z_1, y_1), (z_2, y_2), \dots, (z_{r}, y_{r})$ and $(z'_1, y'_1), (z'_2, y'_2), \dots, (z'_{r}, y'_{r})$ where $\vert z_i y_i \vert \neq \vert z_i' y_i' \vert$, while still having an entry for each. Note that this gives at most $\ell^{\vert Z \vert}$ sets where $\vert z_i y_i \vert = \ell$ and thus, letting $n = \vert v_i u_i\vert$ for some tuple $(v_i, u_i)$ in the input set $Z$, we have $O(n^{\vert Z \vert + 1})$ entries in $V$.
We define $V[\{(z_1, y_1), (z_2, y_2), \dots, (z_r, y_r)\}]$ as a Boolean value, being true iff there there is a valid solution to the pairwise disjoint embedding problem for the set $\{(z_1, y_1), (z_2, y_2), \dots,$ $(z_r, y_r)\}$, and false otherwise. We have, as a base case, that $V[\{(\varepsilon, \varepsilon), (\varepsilon, \varepsilon), \dots, (\varepsilon, \varepsilon)\}]$ is true. Before we provide the recurrence relation used to define our dynamic programming algorithm, we first provide a key observation.

\ifpaper
\else
\input{first_letter_is_necessary}
\input{proof_first_letter_is_nessesary}
\fi

Following \Cref{obs:first_letter_is_nessesary}, let $X(\{(z_1, y_1), \dots, (z_r, y_r)\}) = \bigcap_{i \in [r]} \{z_i[1], u_i[1]\}$. Assuming $X(\{(z_1, y_1), (z_2, y_2), \dots, (z_r, y_r)\})$ is non-empty and contains the letter $x_1$, we define the set $T_1 = (z_1', y_1'), (z_2', y_2'), \dots, (z_r', y_r')$ where $(z_i', y_i')$ is either $(z_i[2, \vert z_i \vert], y_i)$, if $z_i[1] = x_1$, or $(z_i, y_i[2, \vert y_i \vert])$ otherwise. Further, if we have $X(\{(z_1, y_1), (z_2, y_2), \dots, (z_r, y_r)\}) = \{x_1, x_2 \}$ we construct the set $T_2 = (\bar{z}_1', \bar{y}_1'), (\bar{z}_2', \bar{y}_2'), \dots, (\bar{z}_r', \bar{y}_r')$ where $(\bar{z}_i', \bar{y}_i')$ is either $(z_i[2, \vert z_i \vert], y_i)$, if $z_i[1] = x_2$, or $(z_i, y_i[2, \vert y_i \vert])$ otherwise. Note that the difference between $T_1$ and $T_2$ is in the first letter of $z_i$ and cardinality of $X(\{(z_1, y_1), (z_2, y_2), \dots,$ $(z_r, y_r)\})$ is at most two by construction. With this notation, we now define our recurrence relation, using $Z$ to denote the set $\{(z_1, u_1), (z_2, u_2), \dots, (z_r, u_r)\}$:

\begin{equation}
    \label{eq:recurence_relation}
    V[Z] = \begin{cases}
        \text{True,} & \text{if }Z = \{(\varepsilon, \varepsilon)\}, (\varepsilon, \varepsilon), \dots, (\varepsilon, \varepsilon),\\
        \text{False,} & \text{if }X(Z) = \emptyset,\\
        V[T_1] & \text{if }X(Z) = \{x\},\\
        V[T_1] \land V[T_2] & \text{if }X(Z) = \{x_1, x_2\}.
    \end{cases}
\end{equation}

\ifpaper
\else
\input{recurrence_relation}
\input{proof_recurrence_relation}
\fi

Postponing the runtime analysis for a moment, we get the following result for our third main problem, that is to decide whether there exists (and determine a) $w\in\Sigma^{\ell+m}$ for given $S\subseteq\Sigma^{\ell}$ and $u\in\Sigma^m$.

\begin{corollary}
    \label{col:V_table_for_completeness}
    There exists $w \in \Sigma^*$ such that $C(w, u) = \{(v_1, v_2, \dots, v_r)\}$ iff $V[\{(v_1, u), (v_2, u), \dots,$ $(v_r, u)\}]$ is true.
\end{corollary}

Before we can prove the runtime of our construction, we need the following.

\ifpaper
\else
\input{computation_cell_time}
\input{proof_computation_cell_time}
\fi


\ifpaper
\else
\input{constructing_w}
\input{proof_contructing_w}
\fi

This leads us to the final theorem of this section, where we presented algorithms for all three combinations of determining information on one of $u$, $w$, and $C(w,u)$ from the other two.

\ifpaper
\else
\input{compliment_to_w}
\input{proof_compliment_to_w}
\fi

\begin{corollary}
    \label{col:general_problem}
    Given a set of tuples $Z =  \{(z_1, y_1), (z_2, y_2), \dots, (z_r, y_r)\}$ such that $\vert z_i y_i \vert = \vert z_j y_j\vert$, for all $i, j \in [r]$ we can determine, in $O(r n^{r + 1})$ if there exists a word $w$ of length $n$ allowing for a disjoint embedding of each tuple in $Z$. Further, if such a word exists, we can determine it in $O(r n^{r + 1})$ time.
\end{corollary}

%% file: correctness_of_P.tex
\begin{proposition}\label{prop:correctness_of_P}
	For $w \in \Sigma^*$, $u \in \ScatFact(w)$, $i \in [|u|+1]$ and $j \in [|w|]$, we have $P[i,j] = C(w[1..j],u[1..i-1])$.
\end{proposition}

%% file: proof_correctness_of_P.tex
\begin{proof}
	We prove this statement inductively on $j$. For the base case consider $j=1$, i.e., we want to show that $P[i,1] = C(w[1],u[1..i-1])$. We have
	\begin{align*}
		P[i,1] = 
		\begin{cases}
			\{w[1]\}, & \text{if } i=1\\
			\{\varepsilon\}, & \text{if } w[1] = u[1]\\
			\emptyset, & \text{otherwise.}
		\end{cases} 
		= C(w[1],u[1..i-1])
	\end{align*}
	For the inductive step we assume that $P[i,j] = C(w[1..j],u[1..i-1])$. We consider three cases as per Algorithm~\ref{alg:computing_C}.
	\begin{description}
		\item[Case ($i-1 >j+1$).] By Algorithm~\ref{alg:computing_C} we have $P[i,j+1] = \emptyset = C(w[1..j+1],u[1..i-1])$ since $|w[1..j+1]| = j+1 < i-1 = |u[1..i-1]|$.
		\item[Case ($i-1 \leq j+1$).] We subdivide this case based on the value of $w[j + 1]$ relative to $u[i - 1]$.
		\begin{description}
			\item[Subcase ($w{[j+1]}= u{[i-1]}$).] Here we have 
			\begin{align*}
				P[i,j+1] &= P[i-1,j] \cup P[i,j] \cdot w[j+1]\\
				&= C(w[1..j],u[1..i-2]) \cup C(w[1..j],u[1..i-1]) \cdot w[j+1]\\
				&\overset{(*)}{=} C(w[1..j+1],u[1..i-1]).
			\end{align*}
			Where $(*)$ can be shown by two inclusions. First, let $v \in  C(w[1..j],u[1..i-2]) \cup C(w[1..j],u[1..i-1])\cdot w[j+1]$ If $v \in  C(w[1..j],u[1..i-2])$ then, as the \nth{$(j+1)$} letter of $w$ and the \nth{$(i-1)$} letter of $u$ match, $v \in C(w[1..j],u[1..i-1])$ . Otherwise, if $v \in C(w[1..j],u[1..i-1])\cdot w[j+1]$ then we obtain $v \in C(w[1..j],u[1..i-1])$ follows since $v$ is the complement scattered factor that contains the letter $u[i-1]$.
			
			Second, let $v \in C(w[1..j],u[1..i-1])$. Then, $v$ can either be formed using the matching letter, i.e., $v \in C(w[1..j],u[1..i-1])\cdot w[j+1]$ or skipping this letter, i.e., $v \in  C(w[1..j],u[1..i-2])$.
			
			\item[Subcase ($w{[j+1]}\neq u{[i-1]}$).] This case follows the same argument as the second part of the union for the previous case as we can only skip the non matching letter in $u$, i.e., use it if $v$.
		\end{description}
	\end{description}
	This concludes the proof.\qed
\end{proof}

%% file: runtime_complement_set.tex
\begin{theorem}\label{prop:runtime_complement_set}
	Given $w \in \Sigma^*$ and $u \in \ScatFact(w)$ we can compute the set $C(w,u)$ in time $O(|w||u|\binom{w}{u})$ and space $O(|u|\binom{w}{u})$.
\end{theorem}

%% file: proof_runtime_complement_set.tex
\begin{proof}
	The correctness of the computation follows by \Cref{prop:correctness_of_P} since we have $P[|w|,|u|+1] = C(w,u)$. To derive the runtime, observe that the value of $P[i, j]$ can be computed from the values of $P[i - 1, j - 1]$ and $P[i, j - 1]$ in $O(\binom{w}{u})$ time. As there are $\vert w \vert \vert u \vert$ entries in the table, we obtain the stated runtime.\qed
\end{proof}

%% file: intersection.tex
\begin{lemma}\label{lem:u_from_w_and_s_u_in_intersection}
    Let $n\in\N$, $w\in\Sigma^n$, and $S \subseteq \Sigma^m$ for some $m\leq n$. For all $u \in \Sigma^{n-m}$ we have $S \subseteq C(w, u)$ iff $u \in \bigcap_{v \in S} C(w, v)$.
\end{lemma}

%% file: PROOFintersection.tex
\begin{proof}
	Observe that, by definition of the complement set, if $v \in C(w, u)$ then $u \in C(w, v)$. We are going to show both directions separately.
	Therefore, let $S \subseteq C(w,u)$. By definition, this implies that for all $v \in S$ we have $v \in C(w,u)$. By the previously stated observation, for all $v \in S$, $u \in C(w,v)$ which implies that $u \in \bigcap_{v \in S} C(w,v)$.
	
	For the second direction, let $u \in \bigcap_{v \in S} C(w,u)$ which implies that for all $v \in S$, $u \in C(w,v)$. By the dame argument as before, we get that for all $v \in S$, $v \in C(w,u)$, i.e., $S \subseteq C(w,u)$.\qed
\end{proof}

%% file: Swuproblem.tex
\begin{theorem}\label{Swuproblem}
	Given a set $S \subseteq \Sigma^m$ and $w\in\Sigma^n$, we can determine if there exists a word $u$ such that $C(w, u) = S$ in $O(n^2\binom{n}{m})$ time and $O(n \binom{n}{m})$ space.
\end{theorem}

%% file: PROOFSwuproblem.tex
\begin{proof}
	Observe that for each word $v \in S$ we can, by Proposition \ref{prop:runtime_complement_set}, determine the set $C(w, v)$ in $O(n \cdot m \binom{n}{m})$ time and $O(m \binom{n}{m})$ space. Therefore, determining the set $\mathcal{U}$ requires $O(\vert S \vert n \cdot m \binom{n}{m})$ time, without requiring any increase in space. Similarly, as we can determine the value of $C(w, u)$ in $O(n \cdot (n - m) \binom{n}{n - m}) = O(n \cdot (n - m) \binom{n}{m})$ time and $O((n - m) \binom{n}{n - m})$ space, and can check if $C(w, u) = S$ in at most $O(\vert S \vert m \log \vert S \vert)$ time, corresponding to checking each word in $S$ against the words in $C(w, u)$ after sorting each set, we get the total time complexity of $O(n \cdot (n - m) \binom{n}{m} + n \cdot m) \binom{n}{m}) = O(n^2\binom{n}{m})$, and space complexity of $O((n - m) \binom{n}{m} + m \binom{n}{m}) = O(n \binom{n}{m})$.\qed
\end{proof}

%% file: first_letter_is_necessary.tex
\begin{lemma}\label{obs:first_letter_is_nessesary}
    Given a set of tuples $Z = \{(z_1, y_1), (z_2, y_2), \dots, (z_r, y_r)\}$ with $\vert z_i y_i\vert = \vert z_j y_j \vert$, $\forall i,j \in [r]$, $r\in\N$, there exists a word $w$ of length $\vert z_i y_i \vert$ allowing a disjoint embedding of each tuple $(z_i, y_i) \in Z$ only if for all $i\in[r]$ we have either $z_i[1] = w[1]$ or $y_i[1] = w[1]$.
\end{lemma}

%% file: proof_first_letter_is_nessesary.tex
\begin{proof}
	Note that as $\vert w \vert = \vert z_i y_i \vert$, every symbol in $w$ must be in the embedding of either $z_i$ or $y_i$. Thus, $w[1]$ must be the first symbol of either $z_i$ or $y_i$.\qed
\end{proof}

%% file: recurrence_relation.tex
\begin{lemma}\label{lem:recurrence_relation}
    \Cref{eq:recurence_relation} correctly defines the table $V$.
\end{lemma}

%% file: proof_recurrence_relation.tex
\begin{proof}
	We consider each case in \Cref{eq:recurence_relation} separately. First, if we have $Z = \{(\varepsilon, \varepsilon), (\varepsilon, \varepsilon), \dots, (\varepsilon, \varepsilon)\}$, then there exists a disjoint embedding of each tuple into the empty word and thus $V[Z]$ is true. Second, following Observation \ref{obs:first_letter_is_nessesary}, if there is no symbol in the set $X(Z)$, then there can not be any word containing every tuple in $Z$ as a disjoint embedding.
	
	Third, if there is some such shared symbol then any word $w$ containing every tuple in $Z$ as a disjoint embedding must start with the symbol $x$. Further, $w$ must contain a pairwise disjoint embedding of every tuple in $T_1$. In the other direction, if there exists a word $w$ into which there exists a pair of disjoint embeddings of $v_i, u_i$, for each tuple in $T_1$, then the word $x w$ allows a pairwise disjoint embedding of every tuple in $Z$. Finally, if there are two symbols in $X$, $x_1$ and $x_2$, then, by the same argument as above, there is a word allowing for a pair of disjoint embeddings of every tuple in $Z$ iff there is either some word starting with $x_1$ and allowing for a pair of disjoint embeddings of every tuple in $T_1$, or some word starting with $x_2$ and allowing for a pair of disjoint embeddings of every tuple in $T_2$.\qed
\end{proof}

%% file: computation_cell_time.tex
\begin{lemma}
    \label{lem:computation_cell_time}
    The table $V$ for the set $Z = \{(v_1, u_1), (v_2, u_2), \dots, (v_r, u_r)\} $ can be computed in $O(r n^{r + 1})$ time, where $\vert Z \vert = r$ and $\vert v_1 u_1 \vert = n$.
\end{lemma}

%% file: proof_computation_cell_time.tex
\begin{proof}
	Observe first that the value of $V[Z]$, where each tuple $(z_i, y_i)$ satisfies $\vert z_i y_i \vert = \ell$ can be computed in $O(\vert Z \vert)$ time assuming the values of $V[Z']$ have been computed for every $Z' = \{(z_1', y_1'), (z_2', y_2'), \dots, (z_r', y_r')\}$ where $\vert z_i' y_i' \vert = \ell - 1$. Therefore, by first determining the value of $V[Z]$ by increasing the total length of the words in the tuples, we can compute the table in $O(\vert Z \vert \cdot \vert V \vert)$ time, where $\vert V \vert$ is the total number of entries in $V$. As the number of such entries is bounded by $n^{\vert Z\vert + 1}$, we have the stated runtime of $O(\vert Z \vert n^{\vert Z \vert})$. \qed
\end{proof}

%% file: constructing_w.tex
\begin{lemma}
    \label{lem:contructing_w}
    Given a set of tuples $Z = \{(z_1, y_1), (z_2, y_2), \dots, (z_r, y_r)\}$ such that $V[Z]$ is true and $\vert z_i y_i \vert = \vert z_j y_j\vert$, for all $i, j \in [r]$, we can reconstruct some word $w$ of length $n = \vert z_i y_i \vert$ such that there exists a pair of disjoint embeddings of each tuple in $Z$ in $O(\vert Z \vert n)$ time.
\end{lemma}

%% file: proof_contructing_w.tex
\begin{proof}
	We achieve this in a manner analogous to \Cref{alg:computing_C}.
	Specifically, starting at $V[Z]$ we choose as the first symbol $x \in X(Z)$ such that the set $T$ formed by removing the symbol $x$ from the front of one of the words in each tuple satisfies $V[T]$ as being true. We repeat this process until each tuple is formed of two empty words. Correctness follows from Lemma \ref{lem:recurrence_relation} and Corollary \ref{col:V_table_for_completeness}.\qed
\end{proof}

%% file: compliment_to_w.tex
\begin{theorem}
    \label{thm:compliment_to_w}
    Given a set of words $V = \{v_1, v_2, \dots, v_r\} \subseteq \Sigma^{\ell}$ for some $n,r,\ell \in \N$ and a single word $u \in \Sigma^{n - \ell}$ we can determine if there exists some word $w \in \Sigma^n$ such that $C(w, u) = V$ in $O(r n^{r + 1})$ time and further, if such a word exists, we can construct it in $O(r n^{r + 1})$ time.
\end{theorem}

%% file: proof_compliment_to_w.tex
\begin{proof}
	Follows from Lemmas \ref{lem:recurrence_relation} and \ref{lem:contructing_w}, and Corollary \ref{col:V_table_for_completeness}.\qed
\end{proof}

%% file: results.tex
In this section we present combinatorial insights on the complement scattered factor set. Since the runtime of our algorithms depends on $\binom{w}{u}$, i.e., the more embeddings $u$ has in $w$ the longer the algorithms need, we are
interested in insights about $C(w,u)$ independent of the number of embeddings. 
We start with a series of observations. First, we observe that in a word containing a letter square there exists a scattered factor that has one respective complement scattered factor but more than one embedding.

\begin{observation}\label{obslettersquare}
	Let $w \in \Sigma^*$ be some word containing a letter square $\ta\ta$ for $\ta\in\Sigma$, i.e., there exist $x,y\in\Sigma^{\ast}$ with $w=x\ta\ta y$. If $u\in\ScatFact(w)$ is of the form $x_u\ta y_u$ for $x_u\in\ScatFact(x)$ and $y_u\in\ScatFact(y)$,
	we have $|C(w,u)|<|E(w,u)|$.  The reason is that the embeddings $e_1,e_2$ of $u$ in $w$ with $e_1(|x_u|+1)=|x|+1$ and
	$e_2(|x_u|+1)=|x|+2$ result in the same complement scattered factor.
\end{observation}


Further, we want to emphasise the simple fact that having a unique embedding does imply a singleton complement set while the converse does not hold.
\begin{observation}
If $u\in\ScatFact(w)$ for some $w\in\Sigma^{\ast}$ has exactly one embedding, then we have $|C(w,u)|=1$.  Notice that the other direction does not hold in general: consider $w=\ta\ta\ta$ and $u=\ta$, then $C(w,u)=\{\ta\ta\}$ but $|E(w,u)|=3$.
\end{observation}

In the following, we use the universality index $\iota(w)$ (and thus the arch factorisation) of a given word
$w\in\Sigma^{\ast}$ to investigate the number complement scattered factors of a given scattered factor shorter 
than $\iota(w)$.

\ifpaper
\else
\input{bin_coef}
\input{proof_bin_coef_on_words_shorter_than_iota}
\fi

Note that the idea of Lemma~\ref{lem:bin_coef_on_words_shorter_than_iota} cannot be used to give any insight on words $|u| > \iota(w)$: the embeddings may belong to the same complement scattered factors. For instance, the complement scattered factors of $u = \mathtt{ab}$ in $w = \mathtt{ababbaba}$ are $\mathtt{ababba}$ (one occurrence), $\mathtt{abbaba}$ (three occurrences), $\mathtt{abbbaa}$ (one occurrence), $\mathtt{bababa}$ (two occurrences) and $\mathtt{babbaa}$ (one occurrence), i.e., $|C(w,u)| = 5 \not \geq 6 = \binom{4}{2}$.

We now specify properties for $w$ and $u$ such that $\vert C(w, u) \vert \geq 2$.

\ifpaper
\else
\input{iota2}
\input{proof_iota2}
\fi

\ifpaper
\else
\input{ufirstmodus}
\input{proof_ufirstmodus}
\fi

Note that the other direction of \Cref{ufirstmodus} does not hold (by contraposition): if $w=\ta\tb\tb\ta$ and $u=\tb$ we have $C(w,u)=\{\ta\tb\ta\}$.
Nevertheless, we can generalise the result to words of a universality index at least three: If we have at least three arches and three letters in $\Sigma$, we can present a condition on $w$ and $u$ such that we have at least two complement scattered factors of $u$ in $w$.

\ifpaper
\else
\input{nontrivialCwu}
\input{proof_nontrivialCwu}

\fi

Note that the condition $|u|<\iota(w)$ in \Cref{nontrivialCwu} is strictly necessary witnessed by any word $w \in \Sigma^*$ that does have an empty rest and $u = \m(w)$ of length $|u| = \iota(w)$ which has only one unique embedding by definition. Thus, we have $|C(w,\m(w))| = 1$ contradicting the claim of \Cref{nontrivialCwu} without the mentioned condition.
Further, \Cref{nontrivialCwu} yields the unique scattered factor of length $\iota(w)-1$ with $|C(w,u)|=1$.

\begin{corollary}\label{fancywordtakingmodus}
	Let $w \in\Sigma^{\ast}$ with $\m(w)[i] = \ar_{i+1}(w)[1]$ for all $i\in[\iota(w)-1]$ and $\r(w)=\varepsilon$. The scattered factor $u = \m(w)[1..|\m(w)|-1]$ is the only scattered factor of length $\iota(w) -1$ such that $|C(w,u)| = 1$. 
\end{corollary}
 The following result utilises the $\alpha$-$\beta$-factorisation introduced in the context of shortest absent scattered factors (SAS) \cite{DBLP:journals/fuin/KoscheKMS22} in \cite{DBLP:journals/tcs/FleischmannHHHMN23}. We make use of the definition of SAS of a given word here: A word $v$ is a SAS of $w$, if $v \notin \ScatFact(w)$ and, $\ScatFact(v) \setminus \{ v \} \subseteq \ScatFact(w)$.

\input{sas_unique_comp_sf}
 \ifpaper
 \else
 \input{proof_sas_unique_comp_sf}
 \fi
%

The following observation captures the case when $w$ contains a repetition $y^k$ and $u$ contains its base $y$.
Consider $w = \mathtt{bananaananas}, u = \mathtt{ana}$. Here, there are two embeddings with the same complement scattered factor $\mathtt{bnanaanas}$ at indices $(2,8,9)$ and $(2,10,11)$ caused by the square $\mathtt{na}$.

\begin{observation}
	Let $w \in \Sigma^*$ and $u \in \ScatFact(w)$. 
	If $w = xy^k z$ and $u = x' y z'$ for $x,z \in \Sigma^*, y \in \Sigma^+, k \in \N$ with $\al(x') \subseteq \al(x)$ and $\al(z') \subseteq \al(z)$ then $|E_=(w,u)| \geq k$.
\end{observation}

%% file: bin_coef.tex
\begin{lemma}\label{lem:bin_coef_on_words_shorter_than_iota}
	For $w \in \Sigma^*, u \in \ScatFact(w)$ with $|u| \leq \iota(w)$ we have $|E(w,u)| \geq \binom{\iota(w)}{|u|}$.
\end{lemma}

%% file: proof_bin_coef_on_words_shorter_than_iota.tex
\begin{proof}
	By the definition of universality every scattered factor $u$ is contained in a $|u|$-universal word. Thus, every $|u|$-universal word $\ar_{i_1}(w) \ar_{i_2}(w) \cdots \ar_{i_{|u|}}(w) \in \ScatFact(w)$ with $i_j \in [\iota(w)], j \in [|u|]$ and $i_1 < i_2 < \cdots < i_{|u|}$ does contain $u$ as a scattered factor. By combinatorial argument there are $\binom{\iota(w)}{|u|}$ possibilities to form such $|u|$-universal words, and thus at least $\binom{\iota(w)}{|u|}$ different embeddings of $u$ in $w$.\qed
\end{proof}

%% file: iota2.tex
\begin{lemma}\label{iota2}
Given $w\in\Sigma^{\ast}$ with $\iota(w)=2$ and let $u\in\letters(w)$, we have $|C(w,u)|=1$ iff
$u=\m(w)[1]$ and $\ar_2(w)=(\m(w)[1])^r x$ for some $x\in\Sigma^{\ast}$ and $r\in\N$ with $\m(w)[1]\not\in\letters(x)$.
\end{lemma}

%% file: proof_iota2.tex
\begin{proof}
	First, let $|C(w,u)|=1$. If $u \neq\m(w)[1]$ then the embeddings taking $u$ from the first arch and letting the second untouched and taking $u$ from the second arch and letting the first one untouced respectively yield two different complement scattered factors. The other direction follows by the fact that $w$ is of the form $y(\m(w)[1])^rx$
	for some $r\in\N$ and $\m(w)[1]$ neither occurs in $y$ nor in $x$.\qed
\end{proof}

%% file: ufirstmodus.tex
\begin{lemma}\label{ufirstmodus}
Let $w\in\Sigma^{\ast}$ with $\iota(w)\geq 2$. Let $u\in\ScatFact(w)$ with $|u|<\iota(w)$.
Then $|C(w,u)|>1$ if $u[1]\neq \m(w)[1]$.
\end{lemma}

%% file: proof_ufirstmodus.tex
\begin{proof}
	Let $u[1]\neq \m(w)[1]$ and suppose that $C(w,u)=\{v\}$. Since $u$ is strictly shorter than
	the universality index of $w$ there exists an embedding $e_1\in E(w,u)$ such that 
	\[
	e_1(1)\leq|\ar_1(w)|\quad\mbox{and}\quad |\ar_{i-1}(w)|+1\leq e_1(i)\leq |\ar_i(w)|
	\]
	for all $i\in[|u|]$, i.e. we choose the letters for $u$ from the first $|u|$ arches. Let $e_{1i}$ be the index of
	$\ar_i(w)$ that is $e_1(i)$ in $w$, thus the position of $e(i)$ relative to $\ar_i(w)$.
	Set $x_i=\ar_i(w)[1..e_{1i}-1]\ar_i(w)[e_{1i}+1..|\ar_i(w)|]$ for all $i\in[|u|]$, i.e. $w$'s first arch without the letter we picked for $u$. 
	By the same argument we can find an embedding $e_2$ which only uses the last $|u|$ arches of $w$. Define $y_i$ similar to $x_i$ for all $\iota(w)-|u|+1\leq i\leq\iota(w)$.

	If $\iota(w)=2$, we have $|u|=1$. Since $v$ is the only complement scattered factor, we know that
	$x_1\ar_2(w),\ar_1(w)y_2$ are prefixes of $v$. Thus we have $x_1[|x_1|]=\m(w)[1]\neq \ar_1(w)[|x_1|]$ since the second last letter of the first arch cannot be the modus letter - a contradiction.
	
	If $\iota(w)>2$, we have $|u|\geq 1$. 
	Since $v$ is the only complement scattered factor,
	we know that not only $x_1x_2$ but also $\ar_1(w)\ar_2(2)$ are prefixes of $v$. Nevertheless by $u[1]\neq \m(w)[1]$
	we get again $x_1[|x_1|]=\m(w)[1]\neq \ar_1(w)[|x_1|]$ - a contradiction.\qed
\end{proof}

%% file: nontrivialCwu.tex
\begin{theorem}\label{nontrivialCwu}
Let $|\Sigma|\geq 3$. 
If $w\in\Sigma^{\ast}$ with $\iota(w)>2$ and $u\in\ScatFact(w)$ with $|u|<\iota(w)$, $u\neq\varepsilon$ then $|C(w,u)|>1$. 
\end{theorem}

%% file: proof_nontrivialCwu.tex
\begin{proof}
First, notice that we have $\iota(w)>1$ since otherwise $|u|$ would be $0$, thus $u=\varepsilon$.
Let $u\in\ScatFact(w)$ with $|u|<\iota(w)$. 
If $u[1]\neq\m(w)[1]$ the claim follows by \Cref{ufirstmodus}. Thus, assume $u[1]=\m(w)[1]$.

Suppose that $C(w,u)=\{v\}$. Since $u$ is strictly shorter than
the universality index there exists an embedding $e_1\in E(w,u)$ such that 
\[
e_1(1)\leq|\ar_1(w)|\quad\mbox{and}\quad |\ar_{i-1}(w)|+1\leq e_1(i)\leq |\ar_i(w)|
\]
for all $i\in[|u|]$, i.e. we pick $u$'s letters from the first $|u|$ arches.  Let $e_{1i}$ be the index of
$\ar_i(w)$ that is $e_1(i)$ in $w$, thus the position of $e(i)$ relative to $\ar_i(w)$.
Set $x_i=\ar_i(w)[1..e_{1i}-1]\ar_i(w)[e_{1i}+1..|\ar_i(w)|]$ for all $i\in[|u|]$, i.e. $w$'s \nth{$i$} arch without the letter we picked for $u$. 
By the same argument we can find an embedding $e_2$ which only uses the last $|u|$ arches of $w$. Define $y_i$ similar to $x_i$ for all $\iota(w)-|u|+1\leq i\leq\iota(w)$.

Since $\iota(w)>2$ and $v$ is the only complement scattered factor,
we know that not only $x_1x_2$ but also $\ar_1(w)\ar_2(2)$ are prefixes of $v$. Notice $|x_i|+1=|\ar_w(i)|$ for all $i\in[|u|]$.
By $u[1]=\m(w)[1]$, we have $x_1=\inner_1(w)$.

Since $\alpha_1\alpha_2$ is a prefix of $\ar_1(w)\ar_2(w)$, we get $x_2[1]=\ar_2(w)[1]$.
Since $x_2$ is obtained from $\ar_2(w)$ by deleting exactly one letter from $\ar_2(w)$, we get the following cases:\\
\textbf{Case 1}: $e_1(2)=|\ar_1(w)|+1$ (we pick the first letter of the second arch for $u$)\\
Then we have $\ar_2(w)[2..|\ar_2(w)|-1]=\ar_2(w)[1..|\ar_2(w)|-2]$. Thus, there exist $\ta,\tb,\tc\in\Sigma$ and $y\in\Sigma^{\ast}$ with $\ar_2(w)=\ta y\tb\tc$ and hence $y\tb=\ta y$.
This implies $\ta=\tb$ and $y=\ta^r$ for some $r\in\N_0$, i.e. $\ar_2(w)=\ta^{r+2}\tc$.  Thus a contradiction since
we have $|\Sigma|\geq 3$.\\
\textbf{Case 2}: $|\ar_1(w)|+1<e_1(2)<|\ar_1(w)\ar_2(w)|$ 
Thus, we pick a later letter of the second arch for $u$ but not the modus letter. This implies that $\m(w)[2]$ is
the last letter of $x_2$ and $\m(w)[2]\in\inner_2(2)$ by $|\Sigma|\geq 3$ which implies $|\ar_2(w)|\geq 3$.\\
\textbf{Case 3}: $e_1(2)=\m(w)[2]$\\
In this case we have again that $|\letters(\ar_2(w))|=2$ - a contradiction.\qed
\end{proof}

%% file: sas_unique_comp_sf.tex
\begin{proposition}
	Let $w \in \Sigma^*$. If $u \in \SAS(w)$, $|\ar_i(w)|_{u[i]} = 1$ for all $i \in [|u|-1]$ and $u[|u|] \notin \al(\r(w))$ then $u[1..|u|-1]$ is a scattered factor of $w$ with $C(w,u[1..|u|-1]) = 1$.
\end{proposition}

%% file: proof_sas_unique_comp_sf.tex
\begin{proof}
	Let $w = \alpha_1 \beta_1 \cdots \alpha_{\iota(w)} \beta_{\iota(w)} \alpha_{\iota(w)+1}$ be the $\alpha$-$\beta$-factorisation of $w$.
	From \cite[Proposition 6]{DBLP:journals/tcs/FleischmannHHHMN23} we know that $u[1] \in \al(\beta_1) \setminus \al(\alpha_1), u[i] \in \al(\beta_i), u[i]u[i+1] \notin \ScatFact_2(\beta_i\alpha_{i+1})$ for all $i \in [|u|-1]$ and $u[|u|] \notin \al(\r(w))$.
	By definition of SAS, we first of all obtain that $u[1..|u|-1] \in \ScatFact(w)$. Now suppose that $u[1..|u|-1]$ has two complement scattered factors, i.e., $|C(w,u[1..|u|-1])| > 1$. Thus, there must exist two embeddings of $u[1..|u|-1]$ in $w$. Since $|u[1..|u|-1]| = \iota(w)$ and $|\ar_i(w)|_{u[i]} = 1$ for all $i \in [|u|-1]$ this implies that there are either two letters of  one of the embeddings within one of the arches or the last letter $u[i]$ also occurs in the rest. Since $u[|u|] \notin \al(\r(w))$, the latter cannot happen. Thus, $u[i] u[i+1] \in \ar_i(w)$ for some $i \in [\iota(w)]$. By \cite[Proposition 6]{DBLP:journals/tcs/FleischmannHHHMN23} we have $u[i] \in \al(\beta_i)$. Since $|\ar_i(w)|_{u[i]} = 1$ there is no other occurrence of $u[i]$ in $\ar_i(w)$. So $u[i] u[i+1] \in \ar_i(w)$ implies that $u[i]u[i+1] \in \ScatFact_2(\beta_i)$, a contradiction to the condition on absent scattered factors that $u[i]u[i+1] \notin \ScatFact_2(\beta_i\alpha_{i+1})$. Thus, we can conclude that $C(w,u[1..|u|-1]) =1$. \qed
\end{proof}

%% file: singletonEandCsets.tex
We continue with an investigation on $|C(w,u)|=1$ and $|E(w,u)|=1$
and show that for any word that does not contain any squares, the complement scattered factor set of scattered factors that have two different embeddings is always of size at least 2.

\ifpaper
\else
\input{squarefree-large-inverseset}
\input{PROOF-squarefree-large-inverseset}
\fi

\ifpaper
\else
\input{characterisationlettersquarefree}
\input{PROOFcharacterisationlettersquarefree}
\fi

Since we now completely characterised the situation where $C(w,u) = 1$ implies $E(w,u)=1$ for letter square free words $w$ and $u \in \ScatFact(w)$, we now continue with the inspection of $w$ that contain letter squares. Since scattered factors $u$ that contain a letter that uses one part of any letter square immediately have at least two embeddings, we obtain the following result. In the following proposition, we inspect a single letter square of $w$ and referring to this occurrence
if we say that $u$ {\em uses} letters of it.

\ifpaper
\else
\input{skiplettersquares}
\input{PROOFskiplettersquares}
\fi

\begin{remark}
	Note that squares of length at least two are not considered in \Cref{prop:skiplettersquares} since in this specific case the precondition on $|C(w,u)|=1$ is not fulfilled for scattered factors whose embeddings do not have the indices of the square in their image.
	This can be justified as follows: Let $w = x y^2 z$ for some $x,y,z \in \Sigma^*$, $|\al(y) \geq 2$ and $y$ primitive. Now let $u \in \ScatFact(w)$ be of the form $w = x' y' z'$ with $x' \in \ScatFact(x), y' \in \ScatFact(y), z' \in \ScatFact(z)$. Let $y[i]$ for $i \in [|y|]$ be the leftmost occurrence of $y'[1]$ in $y$ and $x'' \in C(x,x'), y'' \in C(y,y'), z'' \in C(z,z')$, we obtain that $x'yy'' z', x' y[1..i-1] y[i+1..|y|] y''[1..i-1] y[i] y''[i+1..|y''|] z' \in C(w,u)$. For the sake of contradiction suppose that $|C(w,u)| = 1$. Thus, $x'yy'' z' = x' y[1..i-1] y[i+1..|y|] y''[1..i-1] y[i] y''[i+1..|y''|] z'$ which implies $yy'' = y[1..i-1] y[i+1..|y|] y''[1..i-1] y[i] y''[i+1..|y''|]$. Further, we can rewrite the left side to $yy'' = y[1..i-1] y[i] y[i+1..|y|] y'' [1..i-1] y''[i..|y''|]$ and obtain
	\begin{align*}
	(y[1..i-1]) (y[i] y[i+1..|y|] y'' [1..i-1]) (y''[i..|y''|])\\ 
	= (y[1..i-1]) (y[i+1..|y|] y''[1..i-1] y[i]) (y''[i+1..|y''|])
	\end{align*}
	which reduces to
	\(
	y[i] y[i+1..|y|] y'' [1..i-1] = y[i+1..|y|] y''[1..i-1] y[i].
	\)
	By the famous Lyndon-Schützenberger Theorem we get $y[i+1..|y|] y'' [1..i-1] \in y[i]^*$ and since $y''[1..i-1] = y[1..i-1]$ (by choice of index $i$) we obtain that $y \in y[i]^*$ - a contradiction to $|\al(y)| \geq 2$.
\end{remark}

%% file: squarefree-large-inverseset.tex
\begin{lemma}\label{squarefree-large-inverseset}
	If $w \in \Sigma^*$ is square-free then for every $u \in \ScatFact(w)$ with $|E(w,u)| \geq 2$  we have $|C(w,u)| >1$.
\end{lemma}

%% file: PROOF-squarefree-large-inverseset.tex
\begin{proof}
	Let $w$ be square-free and $u \in \ScatFact(w)$ such that $|E(w,u)| \geq 2$. Now suppose that $|C(w,u)| = 1$. 
	Since there are two embeddings there exist $v_1,\ldots, v_{|u|+1}, v'_1,\ldots, v'_{|u|+1} \in \Sigma^*$ s.t. $w = v_1 u[1] v_2 u[2] v_3 \cdots v_{|u|} u[|u|] v_{|u|+1} = v'_1 u[1] v'_2 u[2] v'_3 \cdots v'_{|u|} u[|u|] v'_{|u|+1}$ with $v_1 v_2 \cdots v_{|u|+1} = v'_1 v'_2 \cdots v'_{|u|+1}$. 
	Since the embeddings are distinct there exits a leftmost index $j \in [|u|+1]$ such that the $v_j \neq v'_j$. W.l.o.g. let $|v_j| < |v'_j|$. Thus, there exists $z \in \Sigma^*$ such that $v_j u[j] z = v'_j$. 
	
	Since $v_1 v_2 \cdots v_{|u|+1} = v'_1 v'_2 \cdots v'_{|u|+1}$ this implies $v_{j+1}[1] = u[j]$ but then $w$ contains the square $u[j]u[j]$ - a contradiction to the precondition on $w$.\qed
\end{proof}

%% file: characterisationlettersquarefree.tex
\begin{proposition}\label{prop:characterisationlettersquarefree}
	For $w \in \Sigma^*$ we have
	$w$ does not contain any letter squares iff every $u \in \ScatFact(w)$ with $C(w,u) =1$ has only one embedding.
\end{proposition}

%% file: PROOFcharacterisationlettersquarefree.tex
\begin{proof}
	First, we show the left-to-right-implication. Let there be no letter squares in $w$. For the sake of contradiction suppose that there exists a $u \in \ScatFact(w)$ with $|C(w,u)| =1$ that has at least two embeddings. 
	
	If the embeddings do not belong to any square, we can apply \Cref{squarefree-large-inverseset} which implies $|C(w,u)| > 1$, a contradiction to the precondition on $|C(w,u)|$.
	
	Thus, assume that the two embeddings belong to a square. So there exist $w',w'' \in \Sigma^*,v \in \Sigma^+$ such that $w = w' vv w''$. Further, there are $j_1, \ldots, j_\ell \in [|v|]$ ascending and $\overline{w'} \in \ScatFact(w'), \overline{w''} \in \ScatFact(w'')$ such that  
	\begin{align*}
		\overline{w'} \cdot v[1..j_1 -1] v[j_1 +1..j_2 -1] \cdots v[j_{\ell-1}+1..j_\ell -1]v[j_{\ell+1}..|v|] \cdot v \cdot \overline{w''},\\ \overline{w'} \cdot v \cdot v[1..j_1 -1] v[j_1 +1..j_2 -1] \cdots v[j_{\ell-1}+1..j_\ell -1]v[j_{\ell+1}..|v|] \cdot \overline{w''}\\ \in C(w,u).
	\end{align*}
	Since $C(w,u) = 1$ both words are equal, i.e., $v[1..j_1 -1] v[j_1 +1..j_2 -1] \cdots v[j_{\ell-1}+1..j_\ell -1]v[j_{\ell+1}..|v|] \cdot v = v \cdot v[1..j_1 -1] v[j_1 +1..j_2 -1] \cdots v[j_{\ell-1}+1..j_\ell -1]v[j_{\ell+1}..|v|]$. 
	We get that $v[j_1] = v[j_1 +1]$ which implies that $v[j_1] v[j_1 +1] = v[j_1 +1]v[j_1 +1]$ and $w$ contains a letter square which is a contradiction.
	
	Now we show the left-to-right-implication. Therefore, assume that every $u \in \ScatFact(w)$ with $C(w,u) = 1$ has only one embedding. For the sake of contradiction suppose that there exists a letter square in $w$, i.e., we can rewrite $w = w' \ta \ta w''$ for $w',w'' \in \Sigma^*, \ta \in \Sigma$. The scattered factor $w' \ta w'' \in \ScatFact(w)$ has $\ta$ as complement scattered factor so $C(w,w'\ta w'') = \{\ta\}$ but two embeddings since any of the two $\ta$ that can be picked. A contradiction fact that any scattered factor with a singleton complement set has only one embedding.\qed
\end{proof}

%% file: skiplettersquares.tex
\begin{proposition}\label{prop:skiplettersquares}
	Let $w \in \Sigma^*$ contain letter squares.
	For every $u \in \ScatFact(w)$ with $|C(w,u)| = 1$ that uses only one letter of some of the letter squares in $w$ we have $|E(w,u)| >1$. Conversely, for every $u' \in \ScatFact(w)$ with $|C(w,u')| = 1$ that uses none or all letters of the resp. letter squares in $w$ we have $|E(w,u')| = 1$.
\end{proposition}

%% file: PROOFskiplettersquares.tex
\begin{proof}
	First, let $u \in \ScatFact(w)$ with $|C(w,u)| = 1$ that uses only one letter of some of the letter squares in $w$ we have $|E(w,u)| >1$. We immediately get that there are two embeddings of $u$, one that uses the first letter of the square and a second one using the second letter of the square.
	
	Second, let $u' \in \ScatFact(w)$ with $|C(w,u')| = 1$ that uses none or all letters of the letter squares in $w$. First consider the case that none of the letter squares are used. Thus, every embedding of $u'$ in $w$ is also given in $\cond(w)$ where $\cond(w)$ is the \emph{condensed form} of $w=x_1^{k_1}x_2^{k_2}\cdots x_{\ell}^{k_{\ell}}\in\Sigma^{\ast}$ with 
	$k_i,\ell\in\N$, $i\in[\ell]$ is defined by $\cond(w)=x_1\cdots 
	x_{\ell}$ under the assumption that $x_j\neq x_{j+1}$ for $j\in[\ell-1]$. Now for $\cond(w)$ Proposition~\ref{prop:characterisationlettersquarefree} applies since $|C(w,u')| = 1$ implies that $|C(\cond(w),u')| = 1$ and we can conclude that there exists exactly one embedding of $u'$ in $\cond(w)$ and thus also in $w$.
	
	Now suppose that $u'$ uses all letters of every letter square. Suppose that $|E(w,u')| > 1$. Thus, there must exist two embeddings of $u'$ in $w$. 
	Consider the embeddings of $\cond(u')$ in $\cond(w)$. By the construction of $u'$, we can conclude that $\cond(u')$ must have two embeddings in $\cond(w)$. Since $\cond(w)$ is letter square free, we can apply \Cref{prop:characterisationlettersquarefree} and conclude that $\cond(u')$ only has one embedding in $\cond(w)$. Since now $u'$ exhausts every letter square of $w$ completely, we also get that $u'$ has only one embedding $w$, a contradiction.\qed
\end{proof}

%% file: selfshuffle.tex
We finish our investigation with the self-shuffle $u\shuffle u$ for a given $u\in\Sigma^{\ast}$  from the perspective of the complement scattered factor set, i.e., $u \in C(w,u)$ for some $w \in \Sigma$. The first observation is that $u\in C(w,u)$ does not imply $|C(w,u)|=1$: consider $w=\mathtt{abaabaaa}\in\mathtt{abaa}\shuffle\mathtt{abaa}$, here the embedding $(1,5,6,7)\in E(w,u)$ leads to $\mathtt{baaa}\in C(w,u)$.

We now continue with a characterisation of words contained in a selfshuffle.

\begin{lemma}\label{greedyapproach}
For $w,u\in\Sigma^{\ast}$ we have
 $w\in u\shuffle u$ iff there exist $e_1=(i_1,\ldots,i_{|u|})$, $e_2=(j_1,\ldots,j_{|u|})\in E(w,u)$ with $e_1=\overline{e}_2$ and\\
 - $2\leq j_1\leq|w|$ is minimal in $w$ such that $w[j_1]=u[1]$,\\
 - for all $r\in[|u|]_{>1}$ we have that $j_{r-1}<j_r\leq |w|$ is minimal in $w$ such that $w[j_r]=u[r]$.
\end{lemma}
\ifpaper
\else
\input{proof_greedyapproach}
\fi

Using the previous characterisation we are now able to decide the membership of $u$ in $C(w,u)$ in linear time.

\begin{proposition}\label{uinCwu}
Given $w\in\Sigma^{\ast}$ and $u\in\ScatFact(w)$ one can decide in time $O(|w|)$ whether $u\in C(w,u)$.
\end{proposition}
\ifpaper
\else
\input{proof_uinCwu}
\fi

%

For $w\in\Sigma^{\ast}$, $u\in\ScatFact(w)$, $v\in C(w,u)$ with $|v|<|u|$, we have $u\not\in C(w,u)$. In the case $|v|>|u|$ it may be that $u$ is a scattered factor of such a $v\in C(w,u)$. Here, we can alter the idea from \Cref{greedyapproach}.

\begin{proposition}\label{greedyapproach2}
Given $w\in\Sigma^{\ast}$ and $u\in\ScatFact(w)$ one can decide in time $O(|w|)$ whether there exists $v\in C(w,u)$ with
$u\in\ScatFact(v)$.
\end{proposition}
\ifpaper
\else
\input{proof_greedyapproach2}
\fi

The following lemma shows that $u$ is the only complement scattered factor of itself in a word $w$ iff $w$ is the perfect shuffle of $u$.

\begin{theorem}\label{thm:perfectshuffle}
Let $w\in\Sigma^{\ast}$ and $u\in\ScatFact(w)$. Then we have
$C(w,u)=\{u\}$  iff $w\in u\perfectshuffle u$.
\end{theorem}
\ifpaper
\else
\input{proof_perfectshuffle}
\fi

%% file: proof_greedyapproach.tex
\begin{proof}
	Firstly, let $w\in u\shuffle u$. We have to prove that $e_1$ and $e_2$ with the conditions of the right hand side are two disjoint and exhaustive embeddings of $u$ in $w$. Since $w\in u\shuffle u$, we have $w[1]=u[1]$ and we can set $i_1=1$. By the choice of
	$j_1$ we have that
	\[
	|w[2..j_1-1]|_{u[1]}=0.
	\]
	Set $i_s=s$ for all $s\leq j_1-1$. Then $w[i_s]=u[s]$ by $w\in u\shuffle u$. Moreover, we have $u[1]=w[j_1]$. Assume now that for a fixed but arbitrary $r-1\in[|u|]_{>1}$ we have
	\begin{itemize}
		\item $w[j_1]w[j_2]\cdots w[j_{r-1}]=u[1..r-1]$ and
		\item $w[1..(j_1-1)]w[(j_1+1)..(j_2-1)]\cdots w[(j_{r-2}+1)..(j_{r-1}-1)]=u[1..(j_{r-1}-r+1)]$.
	\end{itemize}
	By the choice of $j_r$, we have $u[r]\not\in\letters(w[(j_{r-1}+1)..(j_r-1)]$. Since $w\in u\shuffle u$, we have 
	$w[[1..(j_1-1)]w[(j_1+1)..(j_2-1)]\cdots w[(j_{r-1}+1)..(j_{r}-1)]=u[1..j_r-r]$ and we can set $i_s=j_{r-1}+s$ for all $j_{r-1}<s<j_r$. This concludes this direction of the proof.
	
	\medskip
	
	The other direction follows immediately since $e_1,e_2$ are disjoint, exhaustive embeddings of $u$ in $w$.\qed
\end{proof}

%% file: proof_uinCwu.tex
\begin{proof}
	If $u\not\in C(w,u)$, we have $w\not\in u\shuffle u$. This implies either $|w|\neq 2|u|$ which is checkable in $O(|w|)$ time
	or all $v\in C(w,u)$ with $|v|=|u|$ are different from $u$.
	By \Cref{greedyapproach} we can traverse $w$ once and build the embeddings of $u$ and some $v\in C(w,u)$ letter by letter of $w$.
	Therefore, we have three counter: $i$ for $w$, $c1$ for one occurrence of $u$ and $c2$ for the second. For every $i$, where we cannot set some $j_r$, we set the $c_1$ to $i$ and increment $c_1$ and $i$. If the conditions for $j_r$ are fulfilled, we set $c2$ to $i$ and increment both counters. Once we cannot assign one of $w$'s letters to $u$, we know that $u$ is not in $C(w,u)$ (cf. \Cref{algogreedy}).\qed
\end{proof}

\begin{algorithm}
\textbf{Input:} $w \in \Sigma^*$ and $u \in \ScatFact(w)$\\
\textbf{Output:} \texttt{True} if $u\in C(w,u)$ and \texttt{False} otherwise
\begin{lstlisting}[language=Python]
def findEmbeddings(w,u):
    c1=0
    c2=0
    i=0

    while i < len(w):
        if c1 < c2 and u[c1]==w[i]:
            c1 = c1+1
        elif c1 < c2 and u[c1]!=w[i]:
            if u[c2]==w[i]:
                c2 = c2+1
            else:
                return False
        elif u[c2]==w[i]:
            c2 = c2+1
        else:
            return False
        i = i+1
    return True
\end{lstlisting}
\caption{Check whether $u\in C(w,u)$}\label{algogreedy}
\end{algorithm}

%% file: proof_greedyapproach2.tex
\begin{proof}
	We alter \Cref{greedyapproach} in the following way:
	For $w,u\in\Sigma^{\ast}$ we have
	$w\in u\shuffle v$ and $u\in\ScatFact(v)$ iff $e_1=(i_1,\ldots,i_{|u|}),e_2=(j_1,\ldots,j_{|u|})\in E(w,u)$ with 
	\begin{itemize}
		\item $2\leq j_1\leq|w|$ is minimal in $w$ such that $w[j_1]=u[1]$,
		\item for all $r\in[|u|]_{>1}$ we have that $j_{r-1}<j_r\leq |w|$ is minimal in $w$ such that $w[j_r]=u[r]$.
	\end{itemize}
	The only algorithmic difference to \Cref{greedyapproach} is that we do not require $e_2$ to be the complement embedding to $e_1$. Thus, we assign the actual letter of $w$ to the first occurrence of $u$ if $c1<c2$ and $w[i]=u[c1]$, to the second occurrence if $c1>c2$
	and $w[i]=u[c2]$ and we skip this letter if $u[c1]\neq w[i]\neq u[c2]$.\qed
\end{proof}

%% file: proof_perfectshuffle.tex
\begin{proof}
	Let first $1=|C(w,u)|$ and $u\in C(w,u)$. By the definition of $C(w,u)$, we get $w\in u\perfectshuffle u$. 
	Suppose $w\not\in u\perfectshuffle u$. Then there exists $i\in[|w|]$ such that $w[i]\neq u[\lceil \frac{i}{2}\rceil]$, choose $i$ minimal. Let $j\in[|u|]$ with $w[i]=u[j]$. Then $u[1..\lceil \frac{i}{2}\rceil-1]u[j]\neq u[1..\lceil \frac{i}{2}\rceil]$
	is a prefix of a word $v\in C(w,u)$ - a contradiction.
	
	For the second direction, let $w\in u\shuffle u$ and $|u|=n$. Suppose there exists $u\neq v\in C(w,u)$.
	Let $e_1\in E(w,u)$ and $e_2\in E(w,v)$ such that $e_1\dot\cup e_2$. By $w\in u\perfectshuffle u$ and $e_2\not\in E(w,u)$
	we get
	$e_1\not\in\{(e_{11},e_{12},\ldots,e_{1n})|\,\forall i\in[n]:\, e_{1i}\in\{2(i-1)+1,2i\}\}
	$ where the function values are written as a tuple. Thus, there exists $j\in[n]$ such that $e_1(j)\not\in \{2(j-1)+1,2j\}$.
	Choose $j$ minimal. Thus, we have
	{\tiny
		\begin{align*}
			e_1(1)&=1, &e_1(2)&=3, &\ldots,\quad e_1(j-1)&=2(j-2)+1, \\
			e_2(1)&=2, &e_2(2)&=4, &\ldots,\quad e_2(j-1)&=2(j-1), &e_2(j)&=2j-1 & e_2(j+1)&=2j.
		\end{align*}
	}
	Depicted w.r.t.~$w$ we get
	{\tiny
	\begin{center}
		\begin{tabular}{c||c|c|c|c|c|c|c|c|c|c|c|c|c}
			$w$   & $1$ & $2$ & $3$ & $4$ & $\cdots$ & $2j-3$ & $2j-2$ & $2j-1$ & $2j$  & $2j+1$ & $\cdots$ & $2n-1$ & $2n$\\ \hline
			$u$   & $1$ & $1$ & $2$ & $2$ & $\cdots$ & $j-1$  & $j-1$  & $j$    & $j$   & $j+1$  & $\cdots$ & $n$    & $n$\\ \hline
			$e_1$ & $1$ &     & $2$ &     & $\cdots$ & $j-1$  &        &        &       &        &          &        & \\ 
			$e_2$ &     & $1$  &    & $2$ &          &        & $j-1$  & $j$    & $j+1$ &        &          &        &
		\end{tabular}
	\end{center}
	}
	Thus, we have $v[j]=v[j+1]=u[j]$ and hence $u[1..j]u[j]=u[1..j-1]u[j]^2$ is a prefix of $v$. By $e_1(j)\geq 2j+1$, there exist
	$
	2j+1\leq k_1 < k_2 < \ldots < k_{n-j+1}\leq 2n 
	$
	such that $e_1(j+r)=k_{r+1}$ for $r\in [n-j]_0$. Thus, we have $u[j+r]=w[k_{r+1}]$ for $r\in[n-j]_0$.
	By $w[k_1]=u[\lceil\frac{k_1}{2}\rceil]$ and $\lceil\frac{k_1}{2}\rceil>j$, we have
	$u[\lceil\frac{k_1}{2}\rceil]=u[j]$. Note that there exists $r\in [n-j+1]_0$ such that $k_r+1=k_{r+1}$ by the Pigeon Hole Principle.
	
	Choose $t$ maximal such that $(k_i')_{i\in[t]}$ is a strictly monotone, increasing subsequence of $(k_i)_{i\in[n-j+1]}$ such that
	$k_1=k_1'\quad\mbox{and}\quad k_{i+1}'=e_1\left( \left\lceil \frac{k_i'}{2} \right\rceil \right).
	$ This gives us the sequence $s'=\left(\left\lceil\frac{k_i'}{2}\right\rceil\right)_{i\in[t]}$ which inherits the strictly increasing monotonicity. If for all $r\in[n-j]_0]$ there exists $i\in[t]$ with $k_i'=k_r$ then there exists
	$\ta\in\Sigma$ with $u[j..n]=\ta^{n-j+1}$ and thus $u=v$ which is a contradiction. Suppose there exists $p\in[n-j+1]$
	such that for all $i\in[t]$, we have $k_i'\neq k_p$. This leads to a second sequence $\hat{s}$
	$\hat{k}_1=k_p\quad\mbox{and}\quad \hat{k}_{i+1}=e_1\left( \left\lceil \frac{\hat{k}_i}{2} \right\rceil \right)
	$
	for $i\in[\hat{t}]$ for some $\hat{t}\leq 2n$. Since both sequences contain indices of $u$ and $w=u\perfectshuffle u$,
	there exists $\ell,\ell'\leq 2n$ such that $k_\ell'=\hat{k}_{\ell'}$. This argument via $\hat{s}$ can be applied to
	all indices not occurring in the sequence $s'$, we again obtain $u[j..n]=\ta^{n-j+1}$ and thus $u=v$ - again a contradiction.
	\qed
\end{proof}

%% file: conclusion.tex
In this work, we combined the knowledge about scattered factors and the shuffle operator in order to investigate the
complement scattered factor set. Since a scattered factor may have several embeddings in a word, we obtain for each
such scattered factor a (not necessarily new) complement scattered factor. This approach raised three natural problems taking $w\in\Sigma^{\ast}$, $u\in \ScatFact(w)$, and $C(w,u)$ into account: given two of them, can we compute the third.
For all three problems we presented algorithms and since (as expected) the runtime is not desirable for real-life problems,
we investigated $C(w,u)$ also from a combinatorial approach. Using the knowledge about the arch factorisation and the $\alpha$-$\beta$-factorisation, we investigated in which cases we have one or more complement scattered factor. Moreover,
we solved the problem, when $u$ is a complement scattered factor of $u$, and when it is the only complement scattered factor. It remains an open problem to investigate scattered factors that are longer than the universality index of the given word. Additionally, an upper bound for $|C(w,u)|$ or even a closed formula in specific cases remains open.

%% file: appendixproofs.tex
\setcounter{theorem}{5}

\setcounter{theorem}{8}

\setcounter{theorem}{9}

\input{intersection}
\input{proofintersection}
\setcounter{theorem}{10}

\setcounter{theorem}{11}

\setcounter{theorem}{12}

\setcounter{theorem}{14}

\setcounter{theorem}{15}

\setcounter{theorem}{16}

\setcounter{theorem}{19}

\setcounter{theorem}{21}

\setcounter{theorem}{22}

\setcounter{theorem}{23}

\setcounter{theorem}{24}

\setcounter{theorem}{26}

\setcounter{theorem}{28}

\setcounter{theorem}{29}

\setcounter{theorem}{30}

\setcounter{theorem}{32}

\setcounter{theorem}{33}
\input{uinCwu}

\input{proof_uinCwu}
\setcounter{theorem}{34}

 \setcounter{theorem}{35}

%% file: upper_bound_C.tex
\begin{lemma}
	Given $w \in \Sigma^*$ and $u \in \ScatFact(w)$, we have $|C(w,u)|\leq \lceil \frac{|w|}{2} \rceil$.
\end{lemma}

%% file: proof_upper_bound_C.tex
\begin{proof}
	Observation~\ref{obslettersquare} gives us immediately that the combination of $w$ having a letter square and $u$ {\em using} it,
	cannot lead to an upper bound for $|C(w,u)|$. By $|C(w,u)|\leq |E(w,u)|$, we get the upper bound by letter-square-free words of the
	form $\ta\tb_1\ta\tb_2\ta\cdots\ta\tb_{\ell}\ta$ for $\ell\in\N$ and $\tb_i\in\Sigma$ for all $i\in[\ell]$. Thus, we have $|C(w,\ta)|\leq \lceil \frac{|w|}{2} \rceil$. \qed
\end{proof}

%% file: greedyapproach.tex
\begin{lemma}\label{greedyapproach}
For $w,u\in\Sigma^{\ast}$ we have
 $w\in u\shuffle u$ iff $e_1=(i_1,\ldots,i_{|u|}),e_2=(j_1,\ldots,j_{|u|})\in E(w,u)$ with $e_1=\overline{e}_2$ and
 \begin{itemize}
 \item $2\leq j_1\leq|w|$ is minimal in $w$ such that $w[j_1]=u[1]$,
 \item for all $r\in[|u|]_{>1}$ we have that $j_{r-1}<j_r\leq |w|$ is minimal in $w$ such that $w[j_r]=u[r]$.
 \end{itemize}
\end{lemma}

%% file: uinCwu.tex
\begin{proposition}\label{uinCwu}
Given $w\in\Sigma^{\ast}$ and $u\in\ScatFact(w)$ one can decide in time $O(|w|)$ whether $u\in C(w,u)$.
\end{proposition}

%% file: greedyapproach2.tex
\begin{proposition}\label{greedyapproach2}
Given $w\in\Sigma^{\ast}$ and $u\in\ScatFact(w)$ one can decide in time $O(|w|)$ whether there exists $v\in C(w,u)$ with
$u\in\ScatFact(v)$.
\end{proposition}

%% file: perfectshuffle.tex
\begin{theorem}\label{thm:perfectshuffle}
Let $w\in\Sigma^{\ast}$ and $u\in\ScatFact(w)$. Then we have
$1=|C(w,u)|$ and $u\in C(w,u)$ iff $w\in u\perfectshuffle u$.
\end{theorem}